\documentclass[11pt]{article}   	                 % use "amsart" instead of "article" for AMSLaTeX format; if "amsart" 
\usepackage{amsmath,amsthm, amssymb,amsfonts,amscd,mathabx}

\usepackage[colorlinks=true,linkcolor=blue,citecolor=red]{hyperref}
\usepackage{graphicx}                 % \includegraphics[scale =?, width=?\textwidth]{png/pdf/jpg/eps}
\usepackage{wrapfig}                               %  \begin{wrapfigure}[lineheight]{position}{width}
\usepackage[font=footnotesize]{caption}                     % Appearance of Caption, e.g. [font=small,labelfont=bf] 
\usepackage{float}
\usepackage{tikz-cd}
\usepackage{nicematrix,mathtools}
\usepackage{subcaption}

\usepackage{booktabs} % makes tables nicer

\numberwithin{equation}{section}
\newtheorem{theorem}{Theorem}[section]

\newtheorem{remark}[theorem]{Remark}
{\begin{trivlist}\item[]\textbf{Proof#1 }}%
{\qed\end{trivlist}}

 {\begin{trivlist}\item[]\textbf{Acknowledgments }}{\end{trivlist}}

\usetikzlibrary{shapes.geometric,automata}

\tikzcdset{row sep/normal = 0.9ex}
\tikzcdset{column sep/normal = 1.2em}

\usetikzlibrary{calc}
\usetikzlibrary{decorations.markings}
\newcommand{\R}{\mathbb{R}}

\newcommand{\Rmnum}[1]{\uppercase\expandafter{\romannumeral #1\relax}}

\newcommand{\rmO}{\mathrm{O}}

\newcommand{\rmd}{\mathrm{d}}
\newcommand{\rme}{\mathrm{e}}
\newcommand{\rmi}{\mathrm{i}}

\renewcommand{\Re}{\mathrm{Re}}
\renewcommand{\Im}{\mathrm{Im}}

\renewcommand{\leq}{\leqslant}
\renewcommand{\geq}{\geqslant}

\def\eps{\varepsilon}

\newfam\bifam
\font\tenbi=cmmib10 scaled \magstep1 \font\sevenbi=cmmib10 at 11pt
\font\fivebi=cmmib10 at 6pt \textfont\bifam = \tenbi
\scriptfont\bifam = \sevenbi \scriptscriptfont\bifam= \fivebi

\begin{document}

%%%%%%%%%%%%%%%%%%%%%%%%%%%%%%%%%%%%%%%%%%%%%%%%%%%%%%
% Title Page
%%%%%%%%%%%%%%%%%%%%%%%%%%%%%%%%%%%%%%%%%%%%%%%%%%%%%%
\begin{center}
{\fontsize{15}{15}\fontfamily{cmr}\fontseries{b}\selectfont
Instability of anchored spirals in geometric flows
}\\[0.2in]
A. Cortez$^{1}$, N. Li$^{1}$, N. Mihm$^{1}$, A. Xu$^{2}$, X. Yu$^{3}$, and A.  Scheel$^{1}$\\[0.1in] %authors are not listed alphabetically, but may be intentional?
\textit{\footnotesize
$^1$School of Mathematics, University of Minnesota, Minneapolis, 206 Church St SE, MN 55414, USA\\[0.05in]
$^2$Department of Mathematics, Rice University, 6100 Main Street
Houston, TX 77005, USA\\[0.05in] % removal of room number consistent with other colleges. 
$^3$Department of Mathematics and Statistics, Pomona College, 333 N. College Way,
Claremont, CA 91711, USA
}
\end{center}

\begin{abstract}
\noindent 
We investigate existence, stability, and instability of anchored rotating spiral waves in a model for geometric curve evolution. We find existence in a parameter regime limiting on a purely eikonal curve evolution. We study stability and instability theoretically, in the aforementioned limiting regime, and numerically. We find convective  and absolute oscillatory instability, as well as saddle-node bifurcations. Our results in particular shed light on the instability of spiral waves in reaction-diffusion systems caused by an instability of wave trains against transverse modulations.
\end{abstract}

%%%%%%%%%%%%%%%%%%%%%%%%%%%%%%%%%%%%%%%%%%%%%%%%%%%%%%
\section{Introduction}
%%%%%%%%%%%%%%%%%%%%%%%%%%%%%%%%%%%%%%%%%%%%%%%%%%%%%%

Spiral waves are fascinating self-organized structures that arise in a plethora of settings ranging from microscopic slime mold aggregation patterns \cite{MR982387} to arrhythmias in cardiac tissue \cite{krinsky1978mathematical} and to shapes of galaxies \cite{Aschwanden2018}. Beyond their fascinating shape and their occurrence in many different areas, spirals are intriguing and relevant due to their robustness and their impact on collective behavior as pacemakers. Once formed, spiral waves are often very difficult to destroy, a point of concern in cardiac arrhythmias
\cite{agladze2007interaction}. This observed strong robustness and stability is unfortunately not reflected in mathematical results on existence and stability, which are rare and often incomplete. Existence has been understood in oscillatory media, with a gauge symmetry \cite{MR3604605,MR4151207,aguareles2023rigorous,MR4215710,MR588502,MR665385,hetebrij2021parameterization} and more generally near Hopf bifurcations
\cite{scheelspiral}, in excitable media with some caveats (see \cite{bernoff} and references therein), near non-Ising-Bloch instabilities \cite{hagberg_nib}, in phase oscillators \cite{MR3556834,nanth},
and in geometric models (see \cite{li2024anchoredspiralsdrivencurvature} and references therein).
Stability results are rare with practically no complete results available in the setting of an unbounded domain. Some asymptotic results on spectral stability are derived in \cite{MR665385}, and the comparison structure of curvature evolution was exploited in \cite{li2024anchoredspiralsdrivencurvature} to give a linear and nonlinear stability result. Conceptual results on existence, robustness, the relation between unbounded and large bounded domains, as well as properties of the linearization in unbounded and large bounded domains together with lists of possible instabilities were presented in \cite{sandstede2021spiral}.

From the perspective of the work in \cite{sandstede2021spiral}, we are interested here in the effect of \emph{transverse instabilities}. Wave trains emitted by the spiral resemble concentric circles, or planar waves for large distances from the core. The spectrum of the linearization at a spiral is therefore decomposed into the essential spectrum, induced and completely determined by spectral properties of those planar wave trains, and the point spectrum, induced by instabilities originating at a finite distance from the spiral core. Transverse instabilities are induced by instabilities related to the essential spectrum, caused by modulations of the planar wave trains in the direction perpendicular to the direction of propagation. Curiously, however, modulations in this transverse direction at an infinite distance from the center of rotation rotate with infinite frequency, so that this ``essential spectrum'' is in fact invisible, located at $\pm\rmi\infty$. In other, more technical terms, linearizing at spiral waves in a corotating frame yields a linearized operator that generates a strongly continuous semigroup. In case of a transverse instability, the spectrum of the generator can have a non-positive real part, yet the semigroup can grow exponentially, that is, the spectral mapping theorem fails; see  \cite{engelnagel} for background and \cite[Lem.~3.27]{sandstede2021spiral} for a precise statement. The results here explore precisely this curious situation and identify clusters of eigenvalues that actually precede the crossing of the essential spectrum at $\pm\rmi\infty$, reminiscent of the crossing of eigenvalues near edges of the absolute spectrum identified in \cite{ss_curvature}.

The present work presents results on existence, stability, and instability, in a geometric setting pioneered and then further developed in \cite{MIKHAILOV1991379,MR1257848,YAMADA1993153,MR1629103,MR2341216}, generalizing in particular the work in \cite{li2024anchoredspiralsdrivencurvature}, so that both stability and instability are possible. In particular, we prove the existence of both stable and unstable spiral waves in annuli with inner and outer radii $R_\mathrm{i}$ and $R_\mathrm{o}$, respectively, that is, for curves evolving inside $\Omega=\{R_\mathrm{i}<r<R_\mathrm{o}\}\subset \R^2$ and attached to the boundary $\partial\Omega$,
both for $R_\mathrm{o}$ large and $R_\mathrm{o}=\infty$. Our setting is a reduced model where the spiral wave field is reduced to a sharp interfacial curve that evolves according to intrinsic geometric quantities while being anchored with suitable boundary conditions at the boundary $\partial\Omega$. More explicitly, we postulate that a curve evolves with a normal velocity that is given through a function $c(\kappa,\kappa_{ss})$, where $\kappa$ is the (signed) curvature of the planar curve and $\kappa_{ss}$ is its second derivative with respect to arclength. Of course, one could also include higher derivatives in $s$ or even nonlocal operators. We however specifically focus on
\begin{equation}\label{e:geomlaw}
    c=V+D_2\kappa-D_4 \kappa_{ss},
\end{equation}
postulating a linear relationship in $\kappa$ and $\kappa_{ss}$ which simplified calculations and captures the leading-order effects of instability that we focus on. Geometric evolution laws involving derivatives of curvature  have been studied in the literature in many contexts, referred to as elastic curves associated with an energy $\int\kappa^2$ or as surface diffusion; see for instance \cite{MR1888641} for background on curve evolution equations. More relevant for our context are descriptions of fronts near long-wavelength transverse instabilities;  see~\cite{SIVASHINSKY19771177} for a description via modulation theory (the Kuramoto-Sivashinsky equation), \cite{frankel1987nonlinear} for a derivation of geometric evolution in terms of curvatures, including our particular model as a linearized version, equation (4.5) there, 
and \cite{MR2235817} for further references and connections to intrinsic, geometric equations. 

We orient the curve originating at the inner boundary outward, choose a right normal, and define a signed curvature using unit tangent $T(s)$ and normal $V(s)$ as functions of arclength via $\kappa=\langle N,\dot{T}\rangle$; see Figure~\ref{fig:line-tension}.
% fig:line-tension
\begin{figure*}[!b]
    \centering
    \scalebox{.7}{
    \begin{minipage}{0.38\textwidth}
           \begin{tikzpicture}[font=\Large,scale=2]
    \draw (0,0) circle (1cm);
    \draw[Latex-Latex] (0,0) circle (1pt) [fill] -- (-.707,-.707) node [midway,above left,xshift=-10pt,yshift=-10pt] {$R_\mathrm{i}$};
    \draw[domain=1:3,blue,smooth,variable=\r]
      plot ({\r*cos(31.6*(\r-2))},{\r*sin(31.6*(\r-2))});

    \node [circle,fill,inner sep=1pt] (point) at (2,0) {};
    \node [circle,fill,inner sep=1pt] (CoC) at ($(point) + (-.914,.818)$) {};
    \draw [densely dashed] (point) arc [
        start angle = -41.84,
        end angle = -41.84+45,
        radius = 1.227cm
    ] (point) arc [
        start angle = -41.84,
        end angle = -41.84-45,
        radius = 1.227cm
    ];
    \draw[Latex-Latex] (CoC) -- (point) node [midway,above right] {$\frac{1}{\vert\kappa\vert}$};
    \draw[-Latex,red] (point) -- ++($0.5*(.914,-.818)$) node [near end,left,xshift=0pt,yshift=-10pt,black] {$\vec{N}$};
    \draw[-Latex,green!80!black] (point) -- ++($0.5*(.818,.914)$) node [near end,right,xshift=0pt,yshift=-5pt,black] {$\vec{T}$};
\end{tikzpicture}
    \end{minipage}
    \hspace{80pt}
    \begin{minipage}{0.38\textwidth}
        %for use in presentation and writeup. 
\begin{tikzpicture}[
    font=\Large,
    frame/.pic={
        \draw[-Latex, red](0,0)--(0,-.7);
        \draw[-Latex, green!80!black](0,0)--(.7,0);
        \node[text=black] at (0.05,-1){$\vec{N}$};
        \node[text=black] at (1,-0.1){$\vec{T}$};
    },
    frame/.default={}
    ]
    \node (inner circle) [circle,draw,minimum width=2cm] at (0,0) {};
    \node (outer circle) [circle,draw,minimum width=6cm] at (inner circle) {};
    \draw[fill=black] (0,0) circle (1.5pt);

    % \draw[Latex-Latex] (0,0) -- (inner circle.220) node[midway,above,xshift=-8pt,yshift=-1pt] {$R_\mathrm{i}$};
    % \draw[Latex-Latex] (inner circle.220) -- (outer circle.220) node[midway,above,xshift=-20pt,yshift=-10pt] {$R_\mathrm{o}$};

    \draw[Latex-Latex] (0,0) -- (inner circle.220) node[midway,above,xshift=-8pt,yshift=-1pt] {$R_\mathrm{i}$};
    \draw[Latex-Latex] (0,0) -- (outer circle.120) node[midway,above,xshift=-20pt,yshift=5pt] {$R_\mathrm{o}$};

    \useasboundingbox;

    \draw[Latex-] (3.15,-.56) arc[start angle=-10,end angle=10,radius=3.2cm] node[right] at (3.2,0) {$\omega$};
    \draw[blue, thick] (1.000,0.000)--(1.000,0.000)--(1.000,0.000)--(1.000,0.000)--(1.000,0.000)--(1.000,0.000)--(1.000,0.000)--(1.000,0.000)--(1.001,0.000)--(1.001,0.000)--(1.002,0.000)--(1.002,0.000)--(1.003,0.000)--(1.006,0.000)--(1.009,0.000)--(1.012,0.000)--(1.015,0.000)--(1.030,0.000)--(1.045,0.001)--(1.060,0.002)--(1.075,0.003)--(1.098,0.005)--(1.120,0.007)--(1.142,0.010)--(1.164,0.014)--(1.196,0.020)--(1.227,0.028)--(1.259,0.037)--(1.290,0.047)--(1.334,0.063)--(1.378,0.082)--(1.422,0.104)--(1.465,0.129)--(1.512,0.160)--(1.559,0.194)--(1.604,0.232)--(1.649,0.273)--(1.689,0.314) to pic[sloped]{frame}(1.729,0.359)--(1.767,0.406)--(1.803,0.456)--(1.835,0.504)--(1.865,0.553)--(1.894,0.605)--(1.921,0.658)--(1.946,0.713)--(1.969,0.769)--(1.991,0.827)--(2.011,0.888)--(2.030,0.956)--(2.046,1.027)--(2.058,1.102)--(2.067,1.181)--(2.069,1.209)--(2.070,1.238)--(2.071,1.268)--(2.070,1.298)--(2.069,1.329)--(2.067,1.361)--(2.064,1.394)--(2.060,1.428)--(2.055,1.457)--(2.050,1.488)--(2.044,1.518)--(2.037,1.550)--(2.031,1.577)--(2.024,1.605)--(2.016,1.632)--(2.009,1.660)--(2.004,1.677)--(1.999,1.694)--(1.995,1.710)--(1.990,1.726)--(1.986,1.742)--(1.982,1.758)--(1.978,1.773)--(1.974,1.788)--(1.971,1.801)--(1.968,1.814)--(1.966,1.827)--(1.963,1.840)--(1.961,1.850)--(1.960,1.861)--(1.958,1.871)--(1.957,1.881)--(1.956,1.893)--(1.954,1.904)--(1.954,1.914)--(1.953,1.925)--(1.953,1.935)--(1.953,1.946)--(1.953,1.956)--(1.953,1.965)--(1.954,1.977)--(1.954,1.989)--(1.956,2.000)--(1.957,2.010)--(1.960,2.024)--(1.963,2.038)--(1.966,2.051)--(1.970,2.063)--(1.977,2.083)--(1.985,2.101)--(1.994,2.118)--(2.004,2.134)--(2.014,2.149)--(2.026,2.163)--(2.038,2.177)--(2.051,2.189);
\end{tikzpicture}
    \end{minipage}
    }
    \caption{Setup for geometric curve evolution according to \eqref{e:geomlaw}: shown are unit tangent with a parameterization pointing away from the core anchoring point, the rightward oriented normal, the curvature via inscribed circle of radius $1/|\kappa|$ (left). Note that 
    $\kappa<0$ in this picture, so that the normal speed is less than $V$ when $D_2>0$, effectively shortening and smoothing the curve. Inner and outer boundaries (right) need to be equipped with boundary conditions such as contact angles. Note that $\kappa>0$ near the outer boundary leading to acceleration in this regime when $D_2>0$. Curves converge toward Archimedean spirals rotating rigidly with angular frequency $\omega$. }
    \label{fig:line-tension}
\end{figure*}
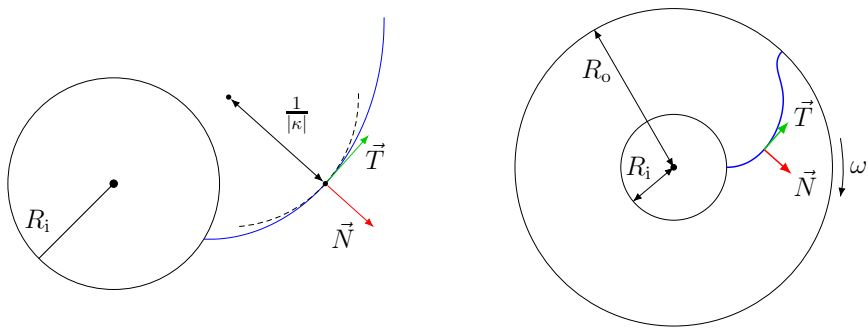

We then specifically focus on  the regime where
\begin{align*}
    &\text{(i)}\ R_\mathrm{i}\gg 1,\  V,D_2,D_4=\rmO(1);\\
    &\text{(ii) } D_4>0\text{ and }  \rmO(1), D_2<0 \text{ and small.} 
\end{align*}
The first condition (i) simply represents a large core, compared to the other quantities $V,D_2,D_4$ in the geometric evolution. Some simple scaling shows that this condition is equivalent to $R_\mathrm{i}=\rmO(1)$, $D_2=\rmO(\eps)$, $D_4=\rmO(\eps^3)$. From this perspective,  in the limiting case $D_2=D_4=0$,  an explicit spiral-wave solution is known and our analysis can be thought of as a singular perturbation of this limiting solution by curvature terms $D_2$ and $D_4$. The second restriction is motivated by the effect of $D_2$ on a planar interface: small perturbations of planar interfaces with $D_2>0$ decay in a diffuse fashion, corresponding to the interpretation of $D_2\kappa$ as a curve-shortening or line-tension effect. Negative line tension, $D_2<0$ then naturally leads to a backward diffusion equation which induces instabilities at all length scales, dampened only at very small scales by the regularizing term $D_4$; see again Figure~\ref{fig:line-tension}.
The anchoring condition naturally leads to a curling-up of the curve, which as a result converges to a rigidly rotating,  (almost) Archimedean spiral so that the long-time dynamics are well described by the existence, stability, and instability of such rotating solutions; see Figure~\ref{fig:filament-winding}.
% fig:filament-winding
\begin{figure*}[!b]
    \centering
     \begin{tikzpicture}[font=\scriptsize]
    \begin{scope}[xshift=-.33\linewidth]
        \draw (0,0) circle (1cm);
        \draw[red] (.866,-.5) -- +(3pt,0) arc [
            start angle = 0,
            end angle = 60,
            radius = 3pt
        ];
        \node[above right,xshift=2pt] at (.866,-.5) {$\vartheta$};
        \draw[blue] (.866,-.5) -- (3,-.5);

        \draw[-Latex] (2,-.5) -- ++(0,-.7) node [midway,right] {$V$};
        \draw[Latex-Latex] (0,0) circle (1pt) [fill] -- (.707,.707) node [midway,above left,yshift=-2pt] {$R_\mathrm{i}$};
    \end{scope}

    \begin{scope}
        \draw (0,0) circle (1cm);
        \draw[red] (0,-1) -- +(1.5pt,-2.6pt) arc [
            start angle = -60,
            end angle = 0,
            radius = 3pt
        ];
        \node[right,xshift=4pt,yshift=-4pt] at (0,-1) {$\vartheta$};
        \draw[blue] (0,-1) to[in=180,out=300] (1,-1.5) -- (3,-1.5);

        \draw[-Latex] (2,-1.5) -- ++(0,-.7) node [midway,right] {$V$};
        \draw[Latex-Latex] (0,0) circle (1pt) [fill] -- (.707,.707) node [midway,above left,yshift=-2pt] {$R_\mathrm{i}$};
    \end{scope}
    
    \begin{scope}[xshift=.33\linewidth]
        \draw (0,0) circle (1cm);
        \draw[red] (-.866,-.5) -- +(-1.5pt,-2.6pt) arc [
            start angle = 240,
            end angle = 300,
            radius = 3pt
        ];
        \node[below,xshift=2pt,yshift=-4pt] at (-.866,-.5) {$\vartheta$};
        \draw[blue] (-.866,-.5) to[in=180,out=240] (1,-3) -- (3,-3);

        \draw[-Latex] (2,-3) -- ++(0,-.7) node [midway,right] {$V$};
        \draw[Latex-Latex] (0,0) circle (1pt) [fill] -- (.707,.707) node [midway,above left,yshift=-2pt] {$R_\mathrm{i}$};
    \end{scope}
    \end{tikzpicture}    \\
    \includegraphics[width=.8\linewidth]{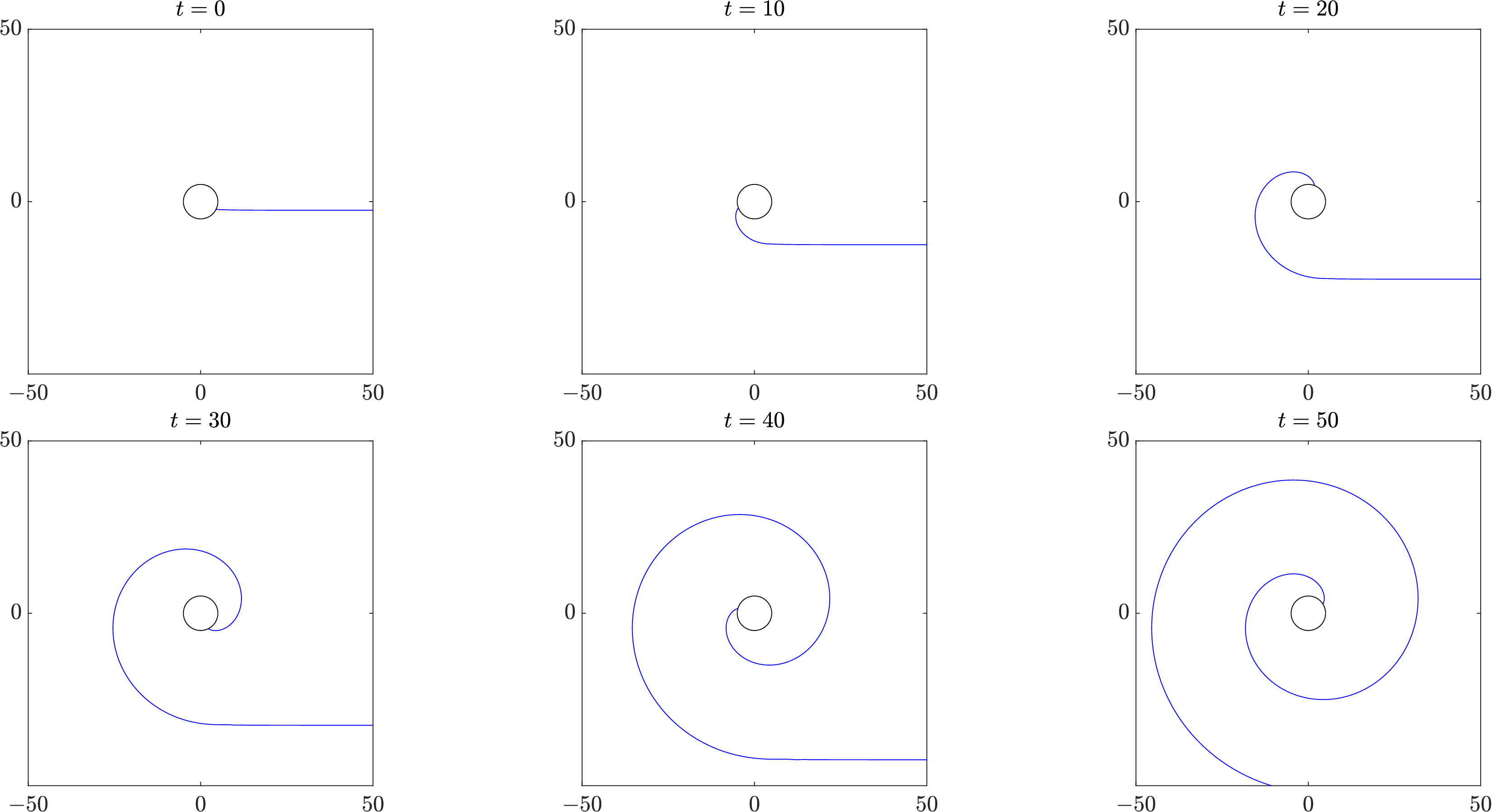}

    \caption{Schematic of the early stages of the evolution of a straight anchored curve into an Archimedean spiral; boundary condition here is a contact angle $\vartheta=\pi/3$ (top and bottom). This schematic is corroborated in direct numerical simulations with $V=D_2=1,D_4=0$ (bottom).
    % The final profile well resembles a best fit Archimedean spiral (red dashed line).
    }
    \label{fig:filament-winding}
\end{figure*}
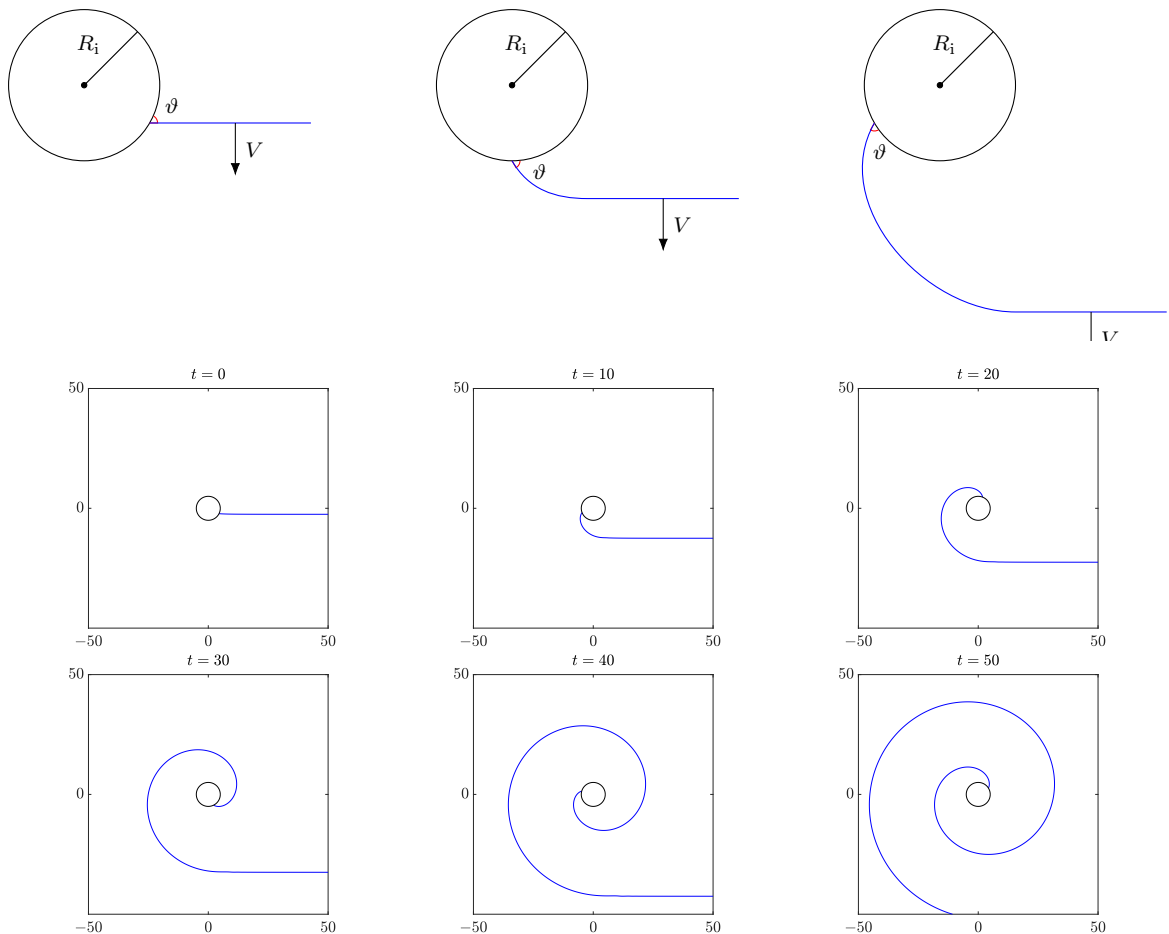

Our main results can be stated informally as follows.

\textbf{Result 1.} Fix $V>0,D_4>0$, and $D_2$. Then for  ``compatible boundary conditions'' at $r=R_\mathrm{i}$ and all $R_\mathrm{i}\gg 1$ sufficiently large, there exists a rigidly rotating spiral wave solution with frequency
\[
\omega=\frac{V}{\sin(\vartheta_\mathrm{i})}R_\mathrm{i}^{-1}+\rmO(R_\mathrm{i}^{-2}),
\]
where $\vartheta_\mathrm{i}$ is the contact angle between the curve and the inner circle as shown in Figure~\ref{fig:filament-winding}.
We refer to Section~\ref{s:3} for a precise characterization of compatible boundary conditions but note here that this is an open class of two conditions on angle, curvature, and arclength derivative of curvature.

\textbf{Result 2.} Fix $V,D_4>0$, and  $R_{\mathrm{i}}\gg1$, and ``compatible boundary conditions'' at $r=R_\mathrm{i}$ as in the first result. Then for $D_2>0$, the linearization at the spiral wave does not possess unstable eigenvalues. On the other hand, when $D_2$ decreases past $D_2^\mathrm{crit}(R_\mathrm{i},D_4,V)<0$, the spiral wave undergoes a Hopf bifurcation with an eigenfunction that grows super-exponentially as $r\to\infty$.

Result 1 is rigorously established using geometric singular perturbation theory. Without rigorous prove, we derive predictions as in Result 2 using similar slow-fast analysis. For both results, we also interpret implications for large bounded domains, when $1\ll R_\mathrm{o}<\infty$.

\textbf{Outline.} The remainder of this paper is organized as follows. We derive explicit nonlinear, degenerate parabolic evolution equations for curves $\gamma$ given as graphs in polar coordinates $\varphi=\Phi(t,r)$ and ODE boundary-value problems for rigidly rotating spirals that solve $\partial_t\Phi=\omega$ in Section~\ref{s:2}. Section~\ref{s:3} states and proves our main results on existence. We introduce rescaled radial variables to transform the fourth-order differential equation for spiral waves into a first-order ODE that is singularly perturbed in the sense of Fenichel's geometric singular perturbation theory and construct solutions as described in Result 1. In the following, Section~\ref{s:4}, we mimic this construction for the eigenvalue problem. We present numerical results both from direct simulation and from numerical continuation in Section~\ref{s:5} and conclude with a discussion in Section~\ref{s:dis}.

\textbf{Acknowledgment.} The authors gratefully acknowledge support through grant NSF-DMS-2205663. Additionally, this material is based upon work supported by the National Science Foundation Graduate Research Fellowship Program under Grant No. 2237827. Any opinions, findings, and conclusions or recommendations expressed in this material are those of the authors and do not necessarily reflect the views of the National Science Foundation.

%%%%%%%%%%%%%%%%%%%%%%%%%%%%%%%%%%%%%%%%%%%%%%%%%%%%%%
\section{Curve evolution in polar coordinates}\label{s:2}
We derive the equations for the evolution of a curve given in polar coordinates $(r,\varphi)$ through $\varphi=\Phi(t,r)$. The calculations here extend the calculations in \cite{li2024anchoredspiralsdrivencurvature} where we considered the case $D_4=0$. For reference, we refer to Figure~\ref{fig:line-tension}.

In the Cartesian plane, the time-dependent parameterized curve is in the form 
\[
\Gamma(t)=\{ \gamma(t,r)|r\geq R_\mathrm{i}\}\subset \R^2,
\]
\[
\gamma(t,r)= \left(\begin{array}{c}r\cos(\Phi(r,t))\\r\sin(\Phi(r,t))\end{array}\right),
\]
with unit tangent and right normal
\[
{T}=\frac{1}{M}\begin{pNiceMatrix}
    \cos{\Phi} - r\Phi_r\sin{\Phi} \\
    \sin{\Phi} - r\Phi_r\cos{\Phi}
\end{pNiceMatrix}\, ,
\]
\[
{N}=\frac{1}{M}\begin{pNiceMatrix}[r]
     \sin{\Phi} - r\Phi_r\cos{\Phi} \\
    -\cos{\Phi} + r\Phi_r\sin{\Phi}
\end{pNiceMatrix}\, ,
\]
\[
M := \frac{\partial s}{\partial r} = \sqrt{1+r^2\Phi_r^2}\, ,
\]
where $M$ is a metric factor induced by the radial (instead of arclength) parameterization.
In order to derive the evolution equation, we compute normal velocity, curvature, and arclength derivatives of curvature, 
\begin{align}
c = &\left\langle {\gamma}_t,{N}\right\rangle=-\frac{r\Phi_t}{M}\,,\nonumber\\
\kappa=&\left\langle\frac{\rmd{T}}{\rmd s},{N}\right\rangle=-\frac{r\Phi_{rr}+r^2\Phi_r^3+2\Phi_r}{M^3}\,,\nonumber\\
\kappa_{ss} =& \frac{\partial r}{\partial s}\frac{\partial}{\partial r}\left[\frac{\partial r}{\partial s}\frac{\partial \kappa}{\partial r}\right]=\frac{1}{M}\frac{\partial}{\partial r}\left[\frac{1}{M}\frac{\partial\kappa}{\partial r}\right]\nonumber\\
=&-\frac{1}{M^9}\bigg(r^5 \Phi_{rrrr} \Phi_{r}^4 + r^5 \Phi_{r}^6 \Phi_{rr} + 15 r^5 \Phi_{r}^2 \Phi_{rr}^3 - 10 r^5 \Phi_{rrr} \Phi_{r}^3 \Phi_{rr}\nonumber\\
&\hspace{4em}+ 3 r^4 \Phi_{r}^7 - 6 r^4 \Phi_{rrr} \Phi_{r}^4 + 21 r^4 \Phi_{r}^3 \Phi_{rr}^2 - 3 r^3 \Phi_{rr}^3 + 2 r^3 \Phi_{rrrr} \Phi_{r}^2\nonumber\\
&\hspace{4em}+ 19 r^3 \Phi_{r}^4 \Phi_{rr} - 10 r^3 \Phi_{rrr} \Phi_{r} \Phi_{rr} + 17 r^2 \Phi_{r}^5- 2 r^2 \Phi_{rrr} \Phi_{r}^2 \nonumber\\
&\hspace{4em}- 33 r^2 \Phi_{r} \Phi_{rr}^2 + r \Phi_{rrrr} + 4 \Phi_{rrr} - 4 \Phi_{r}^3 - 36 r \Phi_{r}^2 \Phi_{rr}\bigg)\,.\label{e:curv}
\end{align}
Substituting these expressions into the governing \eqref{e:geomlaw}, 
we find a fourth-order nonlinear parabolic equation for $\Phi$,
\begin{equation}\label{eq:2.6}
\begin{split}
    \Phi_t &= -\Phi_{rrrr} \frac{D_4}{M^4}  +\Phi_{rrr}\frac{D_4}{M^8}  \bigg(6  r^3 \Phi_r^4+\Phi_{rr} \big(10  r^4 \Phi_r^3+10  r^2 \Phi_r\big)+2 r \Phi_r^2-\frac{4 }{ r}\bigg)\\
    &\hspace{2em} +\Phi_{rr}^3 \frac{D_4}{M^8} \left(-15 r^4 \Phi_r^2+3  r^2\right) +\Phi_{rr}^2 \frac{D_4}{M^8}  \left(-21  r^3 \Phi_r^3+33 r \Phi_r\right)\\
    &\hspace{2em} +\Phi_{rr} \frac{D_4}{M^8}\left(- r^4 \Phi_r^6-19  r^2 \Phi_r^4+36  \Phi_r^2\right) +\Phi_{rr}\frac{D_2}{M^2} - \Phi_r^7 \frac{3 D_4 r^3}{M^8}-\Phi_r^5\frac{17 D_4 r }{M^8}\\
    &\hspace{2em}+\Phi_r^3\frac{\left( D_2 M^6 r^2+4 D_4\right) }{M^8 r}+\Phi_r\frac{2 D_2 }{M^2 r}-\frac{M V}{r}\,.
\end{split}
\end{equation}
For well-posedness on a finite annulus $R_\mathrm{i}<r<R_\mathrm{o}$, we supplement \eqref{eq:2.6} with boundary conditions specifying  fixed contact angles $\vartheta_\mathrm{i/o}$ at $R_\mathrm{i/o}$ and vanishing curvature
\begin{equation}\label{eqn:neumann-bcs}
\begin{aligned}
R_\mathrm{i}\Phi_r(R_\mathrm{i},t) &= \cot{\vartheta_\mathrm{i}}\,, & R_\mathrm{o}\Phi_r(R_\mathrm{o},t) &= \cot{\vartheta_\mathrm{o}}\,,\\
\kappa(R_\mathrm{i},t) &= 0\,,  & \kappa(R_\mathrm{o},t) &= 0\,,
\end{aligned}
\end{equation}
where $\kappa$ is expressed in terms of $R_\mathrm{i/o}$ and derivatives of $\Phi$ as in \eqref{e:curv}. We refer to the end of Section~\ref{s:3} for a discussion of the more general class of boundary conditions allowed in our analysis. 

We are particularly interested in rigidly rotating spiral waves, which are of the form $\Phi(r,t) = \phi(r) - \omega t$ for some constant rotation frequency $\omega > 0$.  The equation for $\phi$ then becomes a fourth-order non-autonomous ordinary differential equation, involving only derivatives of $\phi$, which yields, with variables $\ell_j=\frac{\rmd^j}{\rmd r^j}\phi$, and writing $\left(' = \frac{\rmd}{\rmd r}\right)$,
\begin{align}
    \phi'& = \ell_1\,,\nonumber\\
    \ell_1' &= \ell_2\,,\nonumber\\
    \ell_2' &= \ell_3\,,\nonumber\\
    \ell_3' &= F(r,\ell_1,\ell_2,\ell_3)\,,\label{eq:2.8}
\end{align}
where
\begin{equation}
\begin{split}
F(r,\ell_1,\ell_2,\ell_3)
&= \frac{1}{D_4r}\bigg(M^4r\omega - VM^5 + M^2D_2 \left(r\ell_2 + r^2\ell_1^3+2\ell_1\right)\bigg)\\
&\hspace{1em}-\frac{1}{rM^4}\bigg(r^5\ell_1^6\ell_2+15r^5\ell_1^2\ell_2^3-10r^5\ell_3\ell_1^3\ell_2+3r^4\ell_1^7\\
&\hspace{5em}-6r^4\ell_3\ell_1^4+21r^4\ell_1^3\ell_2^2-3r^3\ell_2^3+19r^3\ell_1^4\ell_2\\
&\hspace{5em}-10r^3\ell_3\ell_1\ell_2+17r^2\ell_1^5-2r^2\ell_3\ell_1^2\\
&\hspace{5em}-33r^2\ell_1\ell_2^2+4\ell_3-4\ell_1^3-36r\ell_1^2\ell_2\bigg)\,,
\end{split}
\end{equation}
and $M(r,\ell_1) = \sqrt{1+r^2\ell_1^2}$.

Note that $\Phi(t,r)$ and $\phi(r)$ do not appear in the right-hand sides of \eqref{eq:2.6} and \eqref{eq:2.8}, a simple consequence of rotational invariance of the anchored curve problem. We may thus consider \eqref{eq:2.8} as a non-autonomous equation for $(\ell_1,\ell_2,\ell_3)$, only. Boundary conditions translate into
\begin{equation}
\begin{aligned}
R_\mathrm{i}\ell_1(R_\mathrm{i}) &= \cot(\vartheta_\mathrm{i})\,, & R_\mathrm{o}\ell_1(R_\mathrm{o}) &= \cot(\vartheta_\mathrm{o})\,, \\
\kappa(R_\mathrm{i}) &= 0\,, & \kappa(R_\mathrm{o}) &= 0 \,,
\end{aligned}
\end{equation}
with
\[
\kappa(r)=-\frac{r\ell_2(r)+r^2\ell_1(r)^3+2\ell_1(r)}{\left(1+r^2\ell_1(r)^2\right)^{3/2}}.
\]
We conclude with the simple eikonal  limiting case $D_2=D_4=0$, when \eqref{eq:2.6} for rigidly rotating spirals simply becomes $\omega=MV/r$, which translates into 
\begin{equation}
    \frac{1}{r^2}+\ell_1^2=\frac{\omega^2}{V^2}. 
\end{equation}
Integrating $\phi'(r)=\ell_1(r)$, we find the \emph{eikonal solutions}
\begin{equation}
\phi(r)=\pm\left[\sqrt{\frac{r^2\omega^2}{V^2}-1}-\arctan\sqrt{{\frac{r^2\omega^2}{V^2}-1}}\right]\,,
\end{equation}
unique up to a translation in $\phi$. The boundary condition then selects the frequency. For instance, setting $\phi'(R_\mathrm{i})=0$ gives $\omega=V/R$ and 
\begin{equation}\label{e:expl}
    \phi(r)=\pm\left[\sqrt{\left(\frac{r}{R_\mathrm{i}}\right)^2-1}-\arctan\sqrt{\left(\frac{r}{R_\mathrm{i}}\right)^2-1}\right]\,;
\end{equation}
compare Figure~\ref{fig:omega-heuristics} for numerical evidence.
Expanding at $r=\infty$ gives 
\begin{equation}
    \phi(r)=\pm\left(\frac{1}{R_\mathrm{i}}r -\frac{\pi}{2}+\frac{R_\mathrm{i}}{2}r^{-1}+\rmO\left(r^{-3}\right)\right),
\end{equation}
and the leading-order term demonstrates that the solution is indeed an asymptotically Archimedean spiral, with distance $2\pi R_\mathrm{i}$ between consecutive arms, that is, a spatial wavenumber $k=1/R_\mathrm{i}=\omega/V$. Inspecting the solutions, only the increasing solution in \eqref{e:expl}, corresponding to the positive sign, yields an outward rotating Archimedean spiral of interest here.  
% fig:omega-heuristics
\begin{figure*}
    \centering
    \includegraphics[width=.85\linewidth]{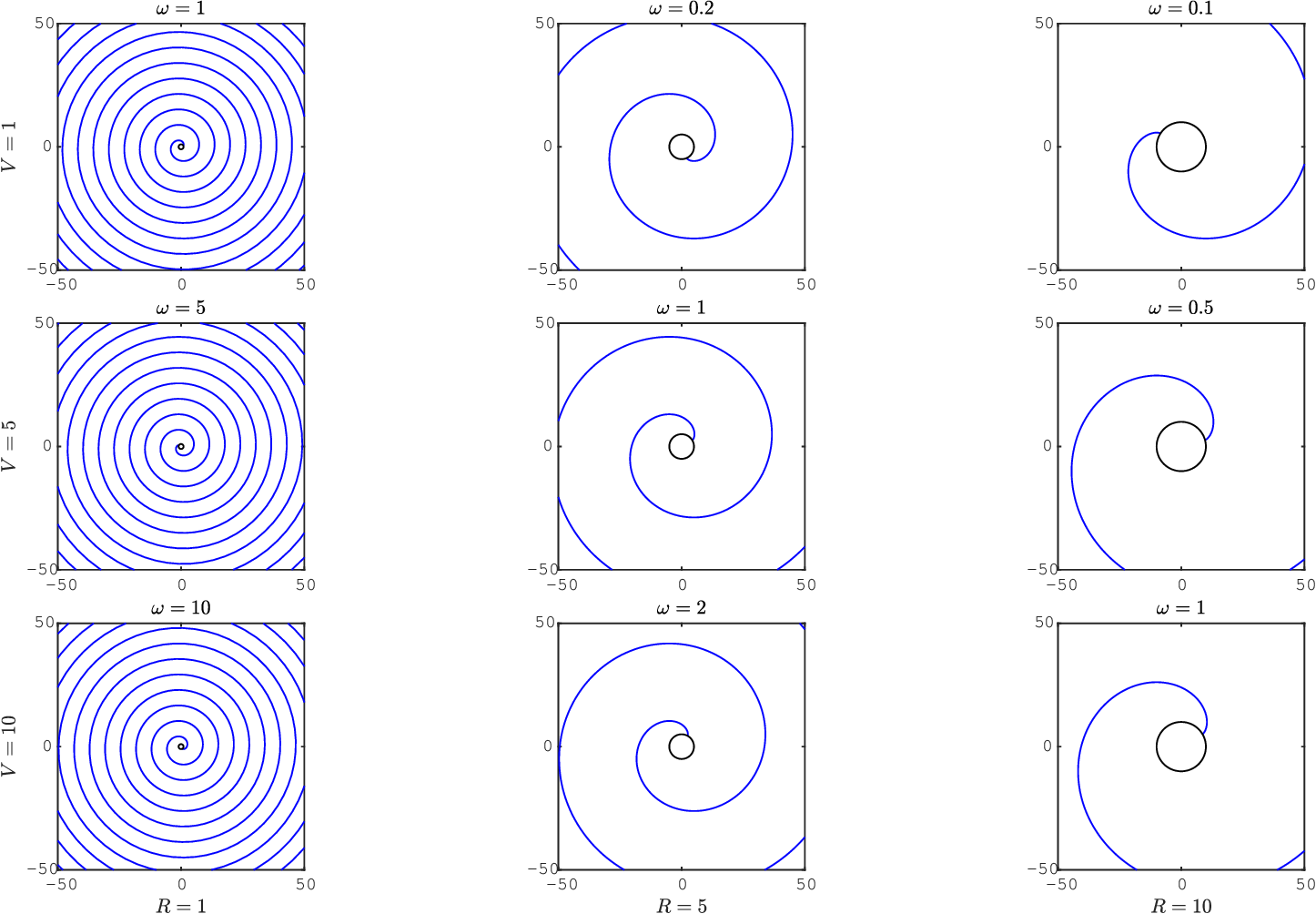}
    \caption{Results of direct simulations of \eqref{eq:2.6} and \eqref{eqn:neumann-bcs}, with $D_2=D_4=0$ and varying $V=1,5,10$ (top to bottom) and $R_\mathrm{i}=1,5,10$ (left to right). The measured frequencies confirm the relation $\omega=V/R$.}
    \label{fig:omega-heuristics}
\end{figure*}

We also note in passing that the asymptotics lack the typical logarithmic correction present in expansions that include $D_2$ \cite[Prop. 3.2]{li2024anchoredspiralsdrivencurvature} or for spiral waves in reaction-diffusion systems \cite[(3.7)]{sandstede2021spiral}.

%%%%%%%%%%%%%%%%%%%%%%%%%%%%%%%%%%%%%%%%%%%%%%%%%%%%%%

\section{The eikonal limit as a geometric singular perturbation theorem}\label{s:3}

In this section, we turn to the analysis of the existence problem for rigidly rotating spirals \eqref{eq:2.8} in the regime where $R_\mathrm{i}=1/\eps\gg 1$. 
% Clearly, this is a singular perturbation  problem since $D_2$ and $D_4$ introduce second and fourth order derivatives into a first-order differential equation. 
Our goal here is to cast this perturbation problem in the framework of Fenichel's geometric singular perturbation \cite{fen} theory, that is, to find variables and time scalings so that in a singular limit $\eps=0$, the system of equations possesses a manifold of equilibria.  

% We start assuming smallness $D_2=\eps\tilde{D}_2$, $D_4=\eps^3\tilde{D}_4$. We then scale $t=\eps\tilde{t}$, and $r=\eps\tilde{r}$ in \eqref{eq:2.6} to see $\eps$ disappear at the expense of a large core $\tilde{R}_\mathrm{i}=R_\mathrm{i}/\eps$. We now drop tildes and proceed in several steps. 

In order to include the limit $r\to\infty$ in our analysis, we introduce the compactified inverse radius $1/r=:\alpha$ as an independent variable, following \cite{li2024anchoredspiralsdrivencurvature,ssradial}, writing
\begin{equation}\label{eq:3.1}
    \begin{cases}
        u = \ell_1\,,\\
        u_1 = \ell_2\,,\\
        u_2 = \ell_3\,,\\
        \alpha = 1/r\,,\\
        \tilde M = \alpha M = \sqrt{\alpha^2 + u^2}\,,
    \end{cases}
    \Longrightarrow
    \begin{cases}
        \frac{\rmd u}{\rmd r} = u_1\,,\\
        \frac{\rmd u_1}{\rmd r} = u_2\,,\\
        \frac{\rmd u_2}{\rmd r} = \tilde F(\alpha,u,u_1,u_2)\,,\\
        \frac{\rmd\alpha}{\rmd r} = -\alpha^2\,,
    \end{cases}
\end{equation}
where
\begin{align*}
\tilde F(\alpha,u,u_1,u_2)
&= \frac{\tilde M^4}{ \alpha^4D_4}\bigg(\omega - V\tilde M + \tilde M^{-2}D_2 \left(u_1\alpha^2 + u^3\alpha+2u\alpha^3 \right)\\
&\hspace{5em}-\alpha^4 D_4\tilde M^{-8}\Big(u^6u_1+15u^2u_1^3-10u_2u^3u_1\\
&\hspace{5em}+3\alpha u^7-6\alpha u_2u^4+21\alpha u^3u_1^2-3\alpha^2u_1^3\\
&\hspace{5em}+19\alpha^2u^4u_1-10\alpha^2u_2uu_1+17\alpha^3u^5\\
&\hspace{5em}-2\alpha^3u_2u^2-33\alpha^3uu_1^2+4u_2\alpha^5\\
&\hspace{5em}-4u^3\alpha^5-36\alpha^4u^2u_1\Big)\bigg)\,.
\end{align*}
Note that $\tilde{F}$ is singular at $\alpha=0$ due to the factor $\alpha^{-4}$. We remedy this by a rescaling of time, and associated weighting of derivatives,
\begin{equation}
    \begin{cases}
        v := u\,,\\
        v_1 := \alpha^{4/3} u_1\,,\\
        v_2 := \alpha^{8/3} u_2\,,\\
        \tau := \frac{3}{7}r^{7/3} \,,
    \end{cases}
    \Longrightarrow
    \begin{cases}
        \frac{\rmd v}{\rmd \tau}   =v_1\,, \\
        \frac{\rmd v_1}{\rmd \tau} = v_2 - \frac{4}{3}\alpha^{7/3} v_1\,, \\
        \frac{\rmd v_2}{\rmd \tau} =  \tilde G(v,v_1,v_2,\alpha) - \frac{8}{3}\alpha^{7/3} v_2\,,\\
        \frac{\rmd \alpha}{\rmd \tau} = -\alpha^{10/3}\,,
    \end{cases}
\end{equation}
where
\begin{align*}
\tilde G(\alpha,v,v_1,v_2)
&= \frac{\tilde M^4}{D_4}\bigg(\omega - V\tilde M + \tilde M^{-2}D_2 \left(v_1\alpha^{2/3} + v^3\alpha+2v\alpha^3 \right)\\
&\hspace{5em}- D_4\tilde M^{-8}\Big(15v^2v_1^3-10v^3v_1v_2-3\alpha^{2}v_1^3\\
&\hspace{5em}-10\alpha^{2}vv_1v_2-6\alpha^{7/3}v^4v_2+21\alpha^{7/3}v^3v_1^2\\
&\hspace{5em}+\alpha^{8/3}v_1v^6-2\alpha^{13/3}v^2v_2-33\alpha^{13/3}vv_1^2\\
&\hspace{5em}+19\alpha^{14/3}v^4v_1+17\alpha^{5}v^5-4\alpha^{5}v^3\\
&\hspace{5em}+3\alpha^{5}v^7+4\alpha^{19/3}v_2-36\alpha^{20/3}v^2v_1\Big)\bigg)\,.
\end{align*}

%  Rescaling 2:
% \begin{equation}
%     \begin{cases}
%         v := u\,,\\
%         v_1 := \alpha^{4/3} u_1\,,\\
%         v_2 := \alpha^{8/3} u_2\,,
%     \end{cases}
%     \quad \Rightarrow \quad
%     \begin{cases}
%         \frac{\rmd v}{\rmd r}   = \alpha^{-4/3}v_1\,, \\
%         \frac{\rmd v_1}{\rmd r} = \alpha^{-4/3}v_2 - \frac{4}{3}\alpha v_1\,, \\
%         \frac{\rmd v_2}{\rmd r} = \alpha^{-4/3} \tilde G(v,v_1,v_2,\alpha) - \frac{8}{3}\alpha v_2\,,\\
%         \frac{\rmd \alpha}{\rmd r} = -\alpha^2\,,
%     \end{cases}
% \end{equation}
% where
% \begin{align*}
% \tilde G(\alpha,v,v_1,v_2) &= \frac{\tilde M^4}{D_4}\bigg(\omega - V\tilde M + \tilde M^{-2}D_2 \left(v_1\alpha^{2/3} + v^3\alpha+2v\alpha^3 \right)\\
% &\hspace{3em}- D_4\tilde M^{-8}\Big(15v^2v_1^3-10v^3v_1v_2-3\alpha^{2}v_1^3-10\alpha^{2}vv_1v_2-6\alpha^{7/3}v^4v_2\\
% &\hspace{3em}+21\alpha^{7/3}v^3v_1^2+\alpha^{8/3}v_1v^6-2\alpha^{13/3}v^2v_2-33\alpha^{13/3}vv_1^2+19\alpha^{14/3}v^4v_1\\
% &\hspace{3em}+17\alpha^{5}v^5-4\alpha^{5}v^3+3\alpha^{5}v^7+4\alpha^{19/3}v_2-36\alpha^{20/3}v^2v_1\Big)\bigg)\,.
% \end{align*}

% Rescaling 3:
% \begin{equation}
%     \tau := \frac{3}{7}r^{7/3} \quad \Rightarrow \quad dr = \alpha^{4/3}\,d\tau \quad \Rightarrow \quad
%     \begin{cases}
%         \frac{\rmd v}{\rmd \tau} = v_1\,,\\
%         \frac{\rmd v_1}{\rmd \tau} = v_2 - \frac{4}{3} \alpha^{7/3} v_1\,,\\
%         \frac{\rmd v_2}{\rmd \tau} = \tilde G(v,v_1,v_2,\alpha) - \frac{8}{3}\alpha^{7/3} v_2\,,\\
%         \frac{\rmd \alpha}{\rmd \tau} = -\alpha^{10/3}\,.\\
%     \end{cases}
% \end{equation}
We remark at this point that the somewhat odd exponents indicate relevant time scales: we will find behavior in $v$ that is exponential in $\tau$, indicating super-exponential growth and decay, for instance in boundary layers, in the original equation with rates $\rme^{\eta r^{7/3}}$ for some $\eta$. 

Our last step scales the size of the core region to $\rmO(1)$, through $\hat{r}=\eps r$,  $\hat\alpha = \varepsilon^{-1}\,\alpha$. The contact angle $\vartheta_\mathrm{i}$ enforces $r v=r\Phi_r\sim\cot\vartheta_\mathrm{i}=\rmO(1)$, so that 
$w=\eps^{-1}v$ is of order 1 near the boundary. Together, we scale and append a trivial equation for the rescaled frequency parameter $\Omega$, 
\begin{equation}\label{eq:3.3}
    \begin{cases}
        w=\varepsilon^{-1}\,v\,,\\
        w_1=\varepsilon^{-7/3}\,v_1\,,\\
        w_2 = \varepsilon^{-11/3} \, v_2\,,\\
        \hat\alpha = \varepsilon^{-1}\,\alpha\,,\\
        T=\varepsilon^{4/3}\tau\,,\\
        \Omega = \omega/\varepsilon\,,\\
        \mathcal{M} = \tilde{M}/\varepsilon\,,
    \end{cases}
    \Longrightarrow 
    \begin{cases}
        \frac{\rmd w}{\rmd T}          = w_1\,, \\
        \frac{\rmd w_1}{\rmd T}        = w_2 - \frac{4}{3}\varepsilon\hat\alpha^{7/3} w_1\,, \\
        \frac{\rmd w_2}{\rmd T}        = \mathcal{G}_{\phantom{2}}-\frac{8}{3}\varepsilon\hat\alpha^{7/3} w_2\,, \\
        \frac{\rmd\hat\alpha}{\rmd T}  = -\varepsilon\hat\alpha^{10/3}\,,\\
        \frac{\rmd\Omega}{\rmd T}=0\,,
    \end{cases}
\end{equation}
where
\begin{align*}
\mathcal G(\hat\alpha,w,w_1,w_2)
&= \frac{\mathcal M^4}{D_4}(\Omega - V\mathcal M) + \frac{\mathcal{M}^2D_2}{D_4} \left(w_1\hat\alpha^{2/3} + \varepsilon w^3\hat\alpha+2\varepsilon w\hat \alpha^3 \right)\\
&\hspace{1em} - \frac{1}{\mathcal M^4}\Big(15w^2w_1^3-10w^3w_1w_2-3\hat\alpha^2w_1^3 - 10\hat\alpha^{2}ww_1w_2\\
&\hspace{1em}-6\varepsilon\hat \alpha^{7/3}w^4w_2+21\varepsilon\hat\alpha^{7/3}w^3w_1^2 - 2\varepsilon\hat\alpha^{13/3}w^2w_2\\
&\hspace{1em} - 33\varepsilon\hat\alpha^{13/3}ww_1^2  + 4\varepsilon\hat\alpha^{19/3}w_2 +  \varepsilon^2\hat\alpha^{8/3}w^6w_1\\
&\hspace{1em} + 19\varepsilon^2\hat\alpha^{14/3}w^4w_1 - 36\varepsilon^2\hat\alpha^{20/3}w^2w_1 + 17\varepsilon^3\hat\alpha^{5}w^5 \\
&\hspace{1em}- 4\varepsilon^3\hat\alpha^{5}w^3+3\varepsilon^3\hat\alpha^5w^7 \Big)\,.
\end{align*}
This is the final form that we were looking for which is regular in $\eps$. Setting $\hat{\alpha}=0$, we find the equilibrium $w_1=w_2=0$, $w=\Omega/V$, for all $\eps$, which corresponds to an asymptotically Archimedean shape. For $\eps=0$, the system reads
\begin{equation}\label{eq:3.4}
\begin{alignedat}{2}
& w'   & &= w_1\,, \\
& w_1' & &= w_2\,, \\
& w_2' & &= \frac{\mathcal{M}^4}{D_4}(\Omega - V\mathcal{M}) + \frac{\mathcal{M}^2D_2}{D_4}w_1\hat\alpha^{2/3} + \frac{10ww_1w_2}{\mathcal{M}^2} - \frac{3w_1^3}{\mathcal{M}^4}(5w^2+\hat\alpha^2)\,, \\
& \hat\alpha' & &= 0\,, \\
&\Omega' & & =0\,.
\end{alignedat}
\end{equation}
The right-hand side vanishes when $w_1=w_2=0$ and either $\mathcal M = 0$ or $\Omega = V\mathcal M$. The case $\mathcal{M}=0$ forces $w=\hat\alpha=0$, corresponding to the origin as an equilibrium of \eqref{eq:3.4}, not a point of interest here. The case $\Omega=V\mathcal M$ gives
\[
\mathcal M^2= w^2+\hat\alpha^2 \stackrel{!}{=} \left(\frac{\Omega}{V}\right)^2 \,,
\]
which yields a manifold of equilibria forming a quarter-circle in the positive quadrant of the $w$-$\hat\alpha$ plane. Note that undoing the scaling introduces a factor  $\varepsilon^2$ and recovers the ODE
\[\left(\frac{\omega}{V}\right)^2 = \phi_r^2+\frac{1}{r^2}\,,\]
whose solution is identical to that of the $D_2=D_4=0$ case briefly explored at the end of Section~\ref{s:2}. In other words, we recover the geometric evolution problem with only normal velocity $V$, the eikonal flow, without line tension $D_2$ and higher-order terms $D_4$, in this particular scaling limit:  The singular perturbation problem of adding higher derivatives is turned into a regular perturbation problem in which the effects of higher derivatives induce a slow flow on a manifold of equilibria. One could of course also scale the planar curve that evolves according to \eqref{e:geomlaw}, magnifying for instance by a factor $1/\eps$, which then induces factors $\eps^{-2}$ and $\eps^{-4}$ in front of $D_2$ and $D_4$, respectively, and after some manipulations arrive at the system here. 

% This suggests that on the slow time scale the source term $V$ dominates over the curvature terms $D_2\kappa$ and $-D_4\kappa_{ss}$. In other words, the far-field limit is an eikonal limit in which the effects of higher derivatives diminish significantly.

% \begin{figure}
% \centering
% \input{Figures/quarter_circle}
% \caption{Diagrammatic representation of the quarter-circle of equilibria in the $w$-$\hat\alpha$ plane.}
% \end{figure}
The boundary conditions in the new scaled variables are two-dimensional manifolds $\mathcal{B}_\mathrm{i/o}$ in the  phase space $(w,w_1,w_2,\hat{\alpha},\Omega)\in\R^5$, given through 
\begin{align}
\mathcal{B}_\mathrm{i}&=\left\{(w,w_1,w_2,\hat{\alpha},\Omega)\,\middle\vert\, 
\begin{aligned}
    &\hat{\alpha}=\hat{\alpha}_\mathrm{i},\ w=\hat{\alpha}_\mathrm{i}\cot(\vartheta_\mathrm{i}),\ \\
    &w_1=-\eps\left(\hat{\alpha}_\mathrm{i}^{1/3}w^3+2\hat{\alpha} _\mathrm{i}^{7/3}w\right)
\end{aligned}\right\}, \label{e:scaled_bdyi}\\
\mathcal{B}_\mathrm{o}&=\left\{(w,w_1,w_2,\hat{\alpha},\Omega)\,\middle\vert\, 
\begin{aligned}
    &\hat{\alpha}=\hat{\alpha}_\mathrm{o},\ w=\hat{\alpha}_\mathrm{o}\cot(\vartheta_\mathrm{o}),\ \\
    &w_1=-\eps\left(\hat{\alpha}_\mathrm{o}^{1/3}w^3+2\hat{\alpha} _\mathrm{o}^{7/3}w\right)
\end{aligned}
\right\}. \label{e:scaled_bdyo}
\end{align}
The differential equation \eqref{eq:3.3} defines a dynamical system $\Psi_T$ on the 5-dimensional phase space. Finding rotating solutions in an annulus then amounts to finding values of $T$ for which $\Psi_T(\mathcal{B}_\mathrm{i})\cap \mathcal{B}_\mathrm{o}\neq \emptyset$. 

For the problem with $R_\mathrm{o}=\infty$, we are interested in solutions that converge to one of the equilibria $\hat{\alpha}=0, w=\Omega/V, w_1=w_2=0$.  As we shall verify below, these equilibria possess a three-dimensional center-stable manifold $W^\mathrm{cs}$. Moreover, the equilibria are asymptotically stable within this (therefore unique) manifold, when restricting to  $\hat{\alpha}\geq 0$, for any $\eps>0$. Spiral waves then correspond to intersections $W^\mathrm{cs}\cap \mathcal{B}_\mathrm{i}\neq \emptyset.$
\begin{theorem}\label{t:ex_infinite}
    Fix $V, D_4 > 0$,  $D_2\in\R$ and consider the system \eqref{eq:3.3} with boundary conditions \eqref{e:scaled_bdyi} and $0<\vartheta_\mathrm{i}<\pi/2$. Then, for all $\eps>0$ sufficiently small, there exists $\Omega_*(\eps)$ and a solution $W(T)=(w,w_1,w_2,\hat{\alpha},\Omega)(T)$ with  
    \[
    W(0)\in\mathcal{B}_\mathrm{i},\quad \lim_{T\to\infty}W(T)=(\Omega_*(\eps)/V,0,0,0,\Omega_*(\eps)).
    \]
    Moreover, we have $\Omega_*(0)=V\hat{\alpha}_\mathrm{i}\csc(\vartheta_\mathrm{i})$
\end{theorem}
Note that by undoing the scalings, this yields a solution in complements of sufficiently large disks of radius $R_\mathrm{i}=\hat{R}_\mathrm{i}/\eps$, with frequency 
\[
    \omega(\eps)=\eps \frac{V}{\hat{R}_\mathrm{i}\sin\vartheta_\mathrm{i}}+\rmO(\eps^2).
\]
The factor $1/\sin\vartheta_\mathrm{i}$ is a simple geometric factor that stems from projecting the speed tangent to the boundary onto the normal; see Figure~\ref{f:proj}. The condition $\vartheta<\pi/2$ allows for $\Phi_r>0$ throughout, avoiding a turnover of the curve and points of high curvature which potentially induce instability. It is not a necessary condition when $D_4=0$, $D_2>0$. However, even in that case, $\vartheta\geq\pi/2$ induces nontrivial boundary layers and stronger corrections  $\rmO(\eps^{5/3})$ to the frequency; see Figure~\ref{f:proj}.
% f:proj
\begin{figure}
    \centering
    \vspace{15pt}\includegraphics[width=0.48\textwidth]{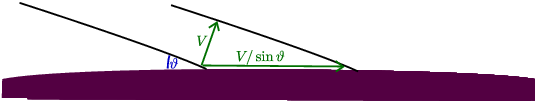}\vspace{15pt}
    \includegraphics[width=0.3\textwidth]{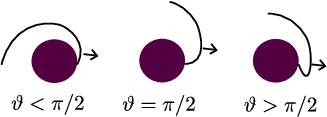}
    \caption{The limiting frequency is determined by the speed along the boundary, which is determined by the fact that its projection on the normal direction equals $V$ (left). acute contact angles allow for monotone $\Phi(r)$, whereas obtuse angles induce turning points and high curvature through necessary boundary layers (right).}\label{f:proj}
\end{figure}
\begin{proof}
Consider the flow of \eqref{eq:3.3} in the five-dimensional phase space $X$. We begin by considering the dynamics in the singular case $\varepsilon=0$, then extend our findings where possible to the nonsingular $0<\varepsilon\ll1$ case.

Accordingly, fix $\varepsilon=0$ and write $w_\mathrm{i}=\hat{\alpha}_\mathrm{i}\cot(\vartheta_\mathrm{i})$. Now choose $\Omega_*>0$ such that $\Omega_*^2=V^2(w_\mathrm{i}^2+\hat\alpha_\mathrm{i}^2)$. The circle segment $\Sigma_0\subseteq X$ given by the equation $\Omega_*^2=V^2(w^2+\hat\alpha^2)$, $w_1=w_2=0$, and the conditions $\hat\alpha\geq0,\,w\geq \delta>0$, forms a smooth manifold of equilibria (it in fact extends smoothly beyond the positive quadrant); see Figure~\ref{fig:nonsingularization} (left).
% fig:nonsingularization
\begin{figure}
    \centering
    \begin{subfigure}[c]{0.34\textwidth}
        \includegraphics[width=\textwidth]{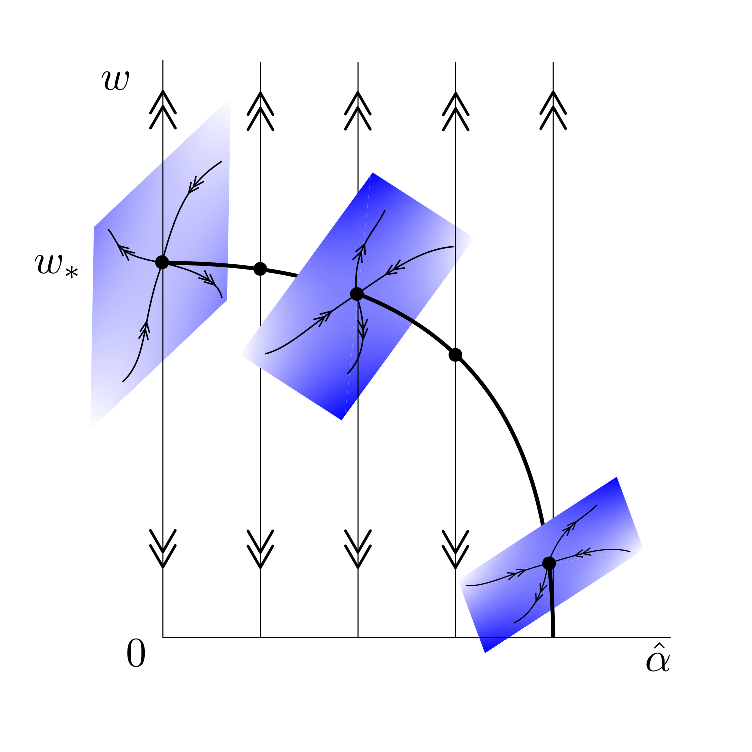}
        % \caption{}
    \end{subfigure}
    \begin{subfigure}[c]{0.26\textwidth}
        \includegraphics[width=\textwidth]{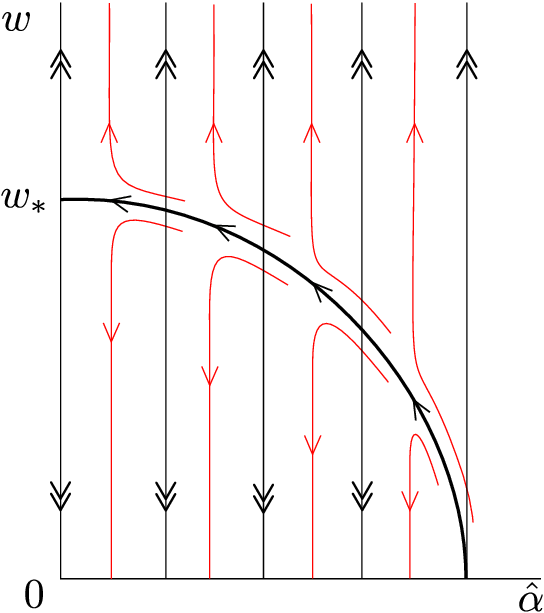}
        % \caption{}
    \end{subfigure}
    
    \caption{Flows of \eqref{eq:3.3} in $\R^4$ for fixed $\Omega$ with $\eps=0$ (left) and $\eps\gtrsim 0$ (right). The vertical hyperplanes $\alpha= const$ are invariant in the left panel and the quarter circle consists entirely of equilibria with 2d strong unstable and 1d strong stable manifolds. We show a projection of unstable fibers and their dynamics onto the $(\hat{\alpha},w)$-plane in the right panel. Strong stable and unstable manifolds continue as invariant foliations (black) for $\eps\gtrsim 0$. Also shown  are sample trajectories (red) off the slow manifold, again projected onto the $(\hat{\alpha},w)$-plane. 
    The boundary condition is a line inside fixed $\hat{\alpha}=1/\hat{R}_\mathrm{i}$.}
    \label{fig:nonsingularization}
\end{figure}

We claim that this manifold $\Sigma_0$ is normally hyperbolic in the sense of Fenichel, that is, the linearization at any equilibrium on $\Sigma_0$ possesses precisely 2 zero eigenvalues, associated with the parameter $\Omega$ and the tangent vector to $\Sigma_0$. Therefore, we compute the Jacobian of \eqref{eq:3.4} at  the point
$(\sqrt{\Omega_*^2/V^2-\hat\alpha^2},0,0,\hat{\alpha},\Omega_*)$,
\begin{equation}\label{eqn:e0-system-jacobian2}
\begin{pNiceMatrix}[columns-width=0pt]
 0 & 1 & 0 & 0 & 0 \\
 0 & 0 & 1 & 0 & 0 \\
 -\dfrac{\Omega_*^3}{D_4V^2}\sqrt{\dfrac{\Omega_*
   ^2}{V^2}-\hat\alpha ^2} & \dfrac{\hat\alpha^{2/3} D_2\Omega_*^2}{D_4V^2} & 0 & -\dfrac{\hat\alpha\Omega_*^3}{D_4V^2} & \dfrac{\Omega_*^4}{D_4V^4} \\ 
 0 & 0 & 0 & 0 & 0 \\
 0 & 0 & 0 & 0 & 0 \\
\end{pNiceMatrix}\,.
\end{equation}
The characteristic polynomial possesses a double root at the origin and three other roots, found as solutions to  \begin{equation}\label{eqn:eigenvalue-polynomial}
0=z^3-V^{14}\hat\alpha^{2/3}\Omega_*^2D_2D_4^7z+V^{22}\Omega_*^3D_4^{11}\sqrt{\frac{\Omega_*^2}{V^2}-\hat\alpha^2}\,.
\end{equation}
We now show that exactly one of these roots has a negative real part and all roots are off the imaginary axis. According to Vieta's formulas, the quadratic coefficient of a cubic polynomial is the sum of its roots and the constant term is the negative of the product of the roots. Since $D_4>0$, the constant term is always positive, so we have either 1 or 3 roots with negative real part. Since the quadratic coefficient is zero, the real parts of the roots sum to zero, so there exists precisely one negative real root which we denote by $z_\mathrm{s}$. The eigenvector associated with $z_\mathrm{s}$ is
\[
\mathbf{e}_s :=
\begin{pmatrix*}
1 & z_\mathrm{s} & z_\mathrm{s}^2 & 0 & 0
\end{pmatrix*}^T\,.
\]
This shows that each point 
$(\sqrt{\Omega_*^2/V^2-\hat\alpha^2},0,0,\hat{\alpha},\Omega_*)\in\Sigma_0$ possesses a two-dimensional center-manifold  $W_{0,\hat\alpha}^c$, a one-dimensional strong stable  $W_{0,\hat\alpha}^s$, and a two-dimensional strong unstable manifold $W_{0,\hat\alpha}^u$. We define the \textit{center-stable manifold} of the system to be the union of the strong stable manifolds through all points,
\[W_0^\mathrm{cs}:=\bigcup_{\hat\alpha\in[0,\Omega_*/V]}\bigcup_{\Omega\sim\Omega_*}(W_{0,\hat\alpha}^c\cup W_{0,\hat\alpha}^s)\,.\]
Clearly, $W_0^\mathrm{cs}$ is 3-dimensional, with tangent space at $\Sigma_0$ given by the sum of the eigenspace associated with the eigenvalue  $z_\mathrm{s}$, and the center eigenspace, spanned by
% None of the $W_{0,\hat\alpha}^s$ or $W_{0,\hat\alpha}^c$, which comprise trajectories in $X$, intersect. It follows that the above union is disjoint and $W_0^\mathrm{cs}$ can be made a manifold by inheriting its constituents' charts. Being the continuous union of one-dimensional manifolds along one parameter, $W_0^\mathrm{cs}$ is two-dimensional.
% 
% Next, notice the kernel of the Jacobian is two-dimensional, since the last two rows are zero rows and the first three rows are linearly independent. Two eigenvectors spanning the kernel (equivalently, the zero eigenspace) are
\[
\mathbf{e}_{\hat\alpha} :=\renewcommand{\arraystretch}{1.5}
\begin{pmatrix*}
\mathllap{-}\hat\alpha\vphantom{\dfrac{\Omega_*}{V^2}} \\ 0 \\ 0 \\ \sqrt{\dfrac{\Omega_*^2}{V^2}-\hat\alpha^2} \\ 0
\end{pmatrix*}
\quad,\quad
\mathbf{e}_{\Omega} :=
\begin{pmatrix*}
\dfrac{\Omega_*}{V^2} \\ 0 \\ 0 \\ 0 \\ \sqrt{\dfrac{\Omega_*^2}{V^2}-\hat\alpha^2}
\end{pmatrix*}\,.\]
% Together with $\mathbf{e}_s$, these vectors span the tangent space of $W_0^\mathrm{cs}$ passing through the point with coordinates~$(\sqrt{\Omega_*^2/V^2-\hat\alpha^2},0,0,\hat{\alpha},\Omega_*)$.
% 
The manifold $\mathcal{B}_\mathrm{i}$ defined by the boundary conditions \eqref{e:scaled_bdyi} is a plane on which $w$, $w_1$, $\alpha$ are constant and $w_2$, $\Omega$ are free. The plane contains the point ${W}_*:=(w_*,0,0,\hat\alpha_*,0)$, hence does intersect $W_0^\mathrm{cs}$. Its tangent space is spanned by the vectors
\begin{equation}\label{e:bdybasis}
\mathbf{b}_1:=\begin{pmatrix*}
0 & 0 & 1 & 0 & 0
\end{pmatrix*}^T
\quad,\quad
\mathbf{b}_2:=\begin{pmatrix*}
0 & 0 & 0 & 0 & 1
\end{pmatrix*}^T.
\end{equation}
We claim that the intersection of the two-dimensional manifold of boundary conditions and the three-dimensional center-stable manifold is transverse. Equivalently, we need to show that the matrix 
$\begin{pNiceMatrix}[vlines]
\mathbf{b}_1 & \mathbf{b}_2 & \mathbf{e}_s & \mathbf{e}_{\hat\alpha} & \mathbf{e}_{\Omega}   
\end{pNiceMatrix}$, 
whose columns span the tangent spaces of $W_0^\mathrm{cs}$ and $\mathcal{B}_\rmi$, given explicitly through
\[
T=\begin{pNiceMatrix}[columns-width = 20pt]
 0 & 0 & 1 & \mathllap{-}\hat\alpha  & \dfrac{\Omega_*}{V^2} \\
 0 & 0 & z_\mathrm{s} & 0 & 0 \\
 1 & 0 & z_\mathrm{s}^2 & 0 & 0 \\
 0 & 0 & 0 & \sqrt{\dfrac{\Omega_*^2}{V^2}-\hat\alpha^2} & 0 \\
 0 & 1 & 0 & 0 & \sqrt{\dfrac{\Omega_*^2}{V^2}-\hat\alpha^2} \\
\end{pNiceMatrix}
\]
has full rank, which in turn follows readily from computing the determinant as $\mathrm{det}(T)=-T_{23}T_{44}T_{15}\neq 0$.

% and determinant
% \[\frac{P^2\Omega_*^2}{V^2(5\vert\Omega_*\vert-4\Omega_*)} \sqrt{\frac{\Omega_*^2}{V^2}-\hat\alpha^2}=\frac{P^2w_*\sqrt{w_*^2+\hat\alpha_*^2}}
% {5-4\operatorname{sgn}{\Omega_*}}\,,\]
% where $\operatorname{sgn}{\Omega_*}$ is the sign of $\Omega_*$. This determinant is zero only when $w_*=0$, indicating that, save perhaps at its boundary, $W_0^\mathrm{cs}$ always intersects $B$ transversely. Recalling that $\hat\alpha'=0$, we note that the projection of $W_0^\mathrm{cs}$ onto the $w$-$\hat\alpha$ plane looks like vertical trajectories converging onto each point of $\Sigma_0$.

% We now turn to consider the case $0<\varepsilon\ll1$. Clearly, ${W}_* :=(\Omega_*/V,0,0,0,\Omega_*)$ is an equilibrium of the nonsingular system.

Fenichel's geometric singular perturbation theory \cite{fen} implies that the manifold $W_0^\mathrm{cs}$ depends smoothly on $\eps$ as a smooth manifold, where smooth here refers to the class $C^k$ for any fixed finite $k$. Since the boundary conditions also depend smoothly on $\eps$, the transverse intersection $W_0^\mathrm{cs}\cap \mathcal{B}_\mathrm{i}$ persists and the locally unique intersection point depends smoothly on $\eps$. Since $W^\mathrm{cs}$, again according to \cite{fen} is smoothly fibered by strong stable foliations over $W_0^\mathrm{c}$, trajectories follow the flow in $W_0^\mathrm{c}$, which simply preserves the parameter $\Omega$ and decreases $\hat{\alpha}\to 0$ so that all solutions in $W_0^\mathrm{cs}$ converge to the equilibrium $W_*$. 
\end{proof}

The last result of this section is concerned with finite annuli, mimicking Theorem~\ref{t:ex_finite}.
\begin{theorem}\label{t:ex_finite}
    Fix $V, D_4 > 0$,  $D_2\in\R$ and consider the system \eqref{eq:3.3} with boundary conditions \eqref{e:scaled_bdyi} and  \eqref{e:scaled_bdyo}, and $0<\vartheta_\mathrm{i}<\pi/2$. Assume a compatible outer contact angle
    \begin{equation}          \label{e:comp}    
R_\mathrm{o}\sin(\vartheta_\mathrm{o})=R_\mathrm{i}\sin(\vartheta_\mathrm{i}).
    \end{equation}
    Then, for all $\eps>0$ sufficiently small, there exists $\Omega_*(\eps)$ and a solution $W(T)=(w,w_1,w_2,\hat{\alpha},\Omega)(T)$ with  
    \[
    W(0)\in\mathcal{B}_\mathrm{i},\qquad 
    W(T)\in\mathcal{B}_\mathrm{o}.
    \]
    Moreover, we have $\Omega_*(0)=V\hat{\alpha}_\mathrm{i}\csc(\vartheta_\mathrm{i})$.
\end{theorem}
\begin{proof}
    We first choose $\Omega_*$ as in the proof of Theorem~\ref{t:ex_infinite}. 
    We then invoke the Exchange Lemma \cite{JK_exchange,JKK_exchange,brunovsky_exchange} to show that $\Phi_{-T}(\mathcal{B}_\mathrm{o})$ converges exponentially to $W_0^\mathrm{cs}$.
    For this, we first note that $\mathcal{B}_\mathrm{o}\cap\Sigma_0\ni (w_\mathrm{o},0,0,\hat{\alpha}_\mathrm{o},\Omega_*)$, due to \eqref{e:comp}. We then claim that the intersection of $\mathcal{B}_\mathrm{o}$ and $W_0^\mathrm{cu}\cap\{\Omega=\Omega_*\}$ is 
    transverse. With this, it follows that $\bigcup_T \Phi_T(\mathcal{B}_\mathrm{o})\cap \{\hat{\alpha}=\hat{\alpha}_\mathrm{i}\}$ is exponentially $\rme^{-\delta/\eps}$-close to $W_0^\mathrm{cs}\cap \{\hat{\alpha}=\hat{\alpha}_\mathrm{i}\}$ and existence of spiral waves follows as in Theorem~\ref{t:ex_infinite}. It remains to check transversality. Basis vectors for the tangent space of $\mathcal{B}_\mathrm{o}$ are  $\mathbf{b}_1$ and $\mathbf{b}_2$ from \eqref{e:bdybasis}. The basis for the unstable eigenspace is spanned by vectors 
    \[
    \mathbf{e}_{1/2} :=
    \begin{pmatrix*}
    1 & {z_\mathrm{s}}_{1/2} & {z_\mathrm{s}}_{1/2}^2 & 0 & 0
    \end{pmatrix*}^T\,.
    \]
    where ${z_\mathrm{s}}_{1/2}$ are the unstable roots of \eqref{eqn:eigenvalue-polynomial}, assuming those are different for now.  We can then proceed as in the proof of Theorem~\ref{t:ex_infinite} and form the matrix of basis vectors to the boundary and the unstable manifold, $\begin{pNiceMatrix}[vlines]
    \mathbf{b}_1 & \mathbf{b}_2 & \mathbf{e}_{1}  & \mathbf{e}_{21}& \mathbf{e}_{\hat\alpha} 
    \end{pNiceMatrix}$, 
    \[
        T=\begin{pNiceMatrix}[columns-width = 25pt]
         0 & 0 & 1 & 1& \mathllap{-}\hat\alpha \\
         0 & 0 & {z_\mathrm{s}}_1 &{z_\mathrm{s}}_2 & 0  \\
         1 & 0 & {z_\mathrm{s}}_1^2 & {z_\mathrm{s}}_2^2& 0 \\
         0 & 0 & 0 & 0&\sqrt{\dfrac{\Omega_*^2}{V^2}-\hat\alpha^2} \\
         0 & 1 & 0 & 0 &0\\
        \end{pNiceMatrix}\,,
    \]
    and readily find that $T$ is invertible, thus showing transversality and concluding the proof. In case of a repeated eigenvalue ${z_\mathrm{s}}_1$ a basis of the unstable eigenspace is $(1,{z_\mathrm{s}}_1,{z_\mathrm{s}}_1^2,0,0)$ and $(0,1,2{z_\mathrm{s}}_1,0,0)$. Substituting this into $T$ again gives an invertible matrix and thus existence.
\end{proof}

\begin{remark}\label{r:90}
We note that in the proof, transversality and normal hyperbolicity fail at $w_*=0$, precluding us from allowing a contact angle $\vartheta=\pi/2$. This case was analyzed in the situation $D_2>0,D_4=0$ in \cite{li2024anchoredspiralsdrivencurvature} by identifying the dynamics at the non-normally hyperbolic point as a slow passage through a saddle-node. Our situation here is significantly complicated in comparison to \cite{li2024anchoredspiralsdrivencurvature}  since in addition to the zero eigenvalue associated with the saddle-node, we simultaneously encounter a crossing of purely imaginary eigenvalues. 
\end{remark}

\begin{remark}\label{r:generalbc}
    From the proof, it is clear that boundary conditions should intersect $W_0^\mathrm{cs}$ transversely. If the intersection occurs in a point on $W_0^\mathrm{c}$, the solution closely follows the trajectory on $W_0^\mathrm{c}$, which is simply the solution discussed at the end of Section~\ref{s:2} for the case $D_2=D_4=0$. Intersections away from $W_0^\mathrm{c}$ induce exponentially localized boundary layers near $r=R_\mathrm{i}$. Understanding the possible existence or even stability of such boundary layers is beyond the scope of this work. 
    
    On the other hand, our arguments are robust, so small changes to boundary conditions will not affect the result. Also, boundary conditions that prescribe $\kappa_s$ instead of $\kappa$ lead to similar results, replacing $\mathbf{b}_1$ by $\mathbf{b}_1=(0,1,0,0,0)$ and again $\mathrm{det}(T)\neq 0$.  
\end{remark}

We also note that the value and in particular the sign of $D_2$ are irrelevant to both Theorem~\ref{t:ex_infinite} and \ref{t:ex_finite}. We did see in numerical experiments that the sign of $D_2$ does affect the existence of boundary layers, that is, it can induce bifurcations when boundary conditions are not close to ``compatible''. More dramatically, we will see in the next section that negative values of $D_2$ do induce instabilities and thereby bifurcations, albeit of an oscillatory nature, which is invisible in the ODE for rigidly rotating solutions. 

%%%%%%%%%%%%%%%%%%%%%%%%%%%%%%%%%%%%%%%%%%%%%%%%%%%%%%
\section{Eigenvalues of the linearization at spiral waves near the eikonal limit}\label{s:4}

In this section, we analyze the effect of negative $D_2$ on the stability of the solutions that we found in Theorems \ref{t:ex_infinite} and \ref{t:ex_finite}. Results in this section follow the geometric ideas laid out thus far and yield predictions for the onset and nature of instability, without rigorously establishing stability or instability in bounded or unbounded geometry. 

For ease of exposition, we first focus on the case of Theorem~\ref{t:ex_finite}, when the equation is strictly parabolic and the linearized operators are simply fourth-order elliptic differential operators in a bounded domain. Towards the end, we shall discuss the implications on stability questions in an unbounded setting, Theorem~\ref{t:ex_infinite}. The key idea of our analysis is that, similar to the existence problem, the equation for eigenfunctions has a slow-fast structure as it inherits the slowly varying coefficients of the existence problem, given by the eikonal solution. The location of eigenvalues in such singularly perturbed problems was analyzed to some extent in \cite{CRS}. The main result there states that eigenvalues accumulate, as $\eps\to 0$, near locations of the absolute spectrum \cite{ssabs}. 

We start our analysis with a slightly modified formulation of the existence problem  \eqref{eq:2.6}.  The calculations here are completely analogous to Section~\ref{s:3}, with however lengthier and more cumbersome expressions. We demonstrate here the sequence of transformations retaining only one term with coefficient $D_4,D_2$, and $V$, respectively, as other terms with these coefficients scale analogously, and give the complete end result in \eqref{e:alllin}.  We set $u:= \phi_r$, and rewrite 
\begin{align}\label{eqn:alt-main-eq}
    u_t  &= \left[-\frac{D_4 u_{rrr}}{(1 + r^2 u^2)^2} + \frac{D_2 u_r}{(1 + r^2 u^2)} - \frac{V}{r}(1 + r^2 u^2)^{1/2} + \ldots\right]_r \nonumber\\
    &=: [m]_r\,,
\end{align} 
where $m = -\frac{M}{r}(V+D_2\kappa - D_4 \kappa_{ss})$, and we omitted most nonlinear terms. We first focus on equilibrium solutions $\overline{u}$, which solve, setting again  $\alpha = \frac{1}{r}$,
\begin{equation}\label{eqn:equilibrium-1}
\begin{alignedat}{2}
& \overline{u}_r & &= \overline{v}\,, \\
& \overline{v}_r & &= \overline{w}\,, \\
& \overline{w}_r & &= -\alpha^{-4} \frac{(\alpha^2 + \overline{u}^2)^2}{D_4}\overline{m} + \alpha^{-2}\frac{D_2\overline{v}}{D_4}(\alpha^2 + \overline{u}^2) - \alpha^{-4}\frac{V}{D_4}(\alpha^2 + \overline{u}^2)^{5/2} + \ldots\,,\\
& \overline{m}_r & &= 0\,.
\end{alignedat}
\end{equation}
Note that the constant of integration $\overline{m}$ was referred to as $\omega$, previously. Again, a first change of variables removes the $\alpha^{-4}$ singularity, setting
\begin{equation}
\begin{cases}
    \overline{u} = \overline{u}_1\\
    \overline{v} = \alpha^{-4/3} \overline{v}_1\\
    \overline{w} = \alpha^{-8/3} \overline{w}_1\\
    \overline{m} = \overline{m}_1\\
    \tau = \frac{3}{7}r^{7/3}\\
\end{cases}
\end{equation}
% $\overline{u} = \overline{u}_1$, $\overline{v} = \alpha^{-4/3} \overline{v}_1$, $\overline{w} = \alpha^{-8/3} \overline{w}_1$, $\overline{m} = \overline{m}_1$, $\tau = \frac{3}{7}r^{7/3}$, 
such that $\frac{{\rmd }\tau}{{\rmd }r} = \alpha^{-4/3}$, to obtain
\begin{equation}
\begin{alignedat}{2}
    & \overline{u}_{1\tau} & &= \overline{v}_1\,, \\
    & \overline{v}_{1\tau} & &= \overline{w}_1\,, \\
    & \overline{w}_{1\tau} & &= \frac{(\alpha^2 + \overline{u}_1^2)^2}{D_4}\overline{m}_1 + \alpha^{2/3}\frac{D_2\overline{v}_1}{D_4}(\alpha^2 + \overline{u}_1^2) - \frac{V}{D_4}(\alpha^2 + \overline{u}_1^2)^{5/2} + \ldots\,,\\
    & \overline{m}_{1\tau} & &= 0\,, \\
    & \alpha_\tau & &= -\alpha^{10/3}\,.
\end{alignedat}
\end{equation}
Next, we linearize \eqref{eqn:alt-main-eq} at this equilibrium and look for solutions $u\rme^{\lambda t}$. Still displaying only sample terms, we find
\begin{equation}
\begin{alignedat}{2}
& u_r & &= v\,, \\
& v_r & &= w\,, \\
& w_r & &= -\alpha^{-4} \frac{(\alpha^2 + \overline{u}^2)^2}{D_4}m + \alpha^{-2}\frac{D_2}{D_4}(\alpha^2 + \overline{u}^2)v - 5\alpha^{-4}\frac{V}{D_4}(\alpha^2 + \overline{u}^2)^{3/2}\overline{u}u + \ldots\,,\\
& m_r & &= \lambda u\,.
\end{alignedat}
\end{equation}
Again removing the $\alpha^{-4}$ singularity,  we implement the  change of variables 
\begin{equation}
    \begin{cases}
        \overline{u} = \overline{u}_1\\
        \overline{v} = \alpha^{-4/3} \overline{v}_1\\
        \overline{w} = \alpha^{-8/3} \overline{w}_1\\
        \overline{m} = \overline{m}_1
    \end{cases}
    \qquad
    \begin{cases}
        u = u_1\\
        v = \alpha^{-4/3} v_1\\
        w = \alpha^{-8/3} w_1\\
        m = m_1\\
        \tau = \frac{3}{7}r^{7/3}
    \end{cases}
\end{equation}
% \[
% \overline{u} = \overline{u}_1,\ \overline{v} = \alpha^{-4/3} \overline{v}_1,\ \overline{w} = \alpha^{-8/3} \overline{w}_1, \  \overline{m} = \overline{m}_1,
% \] 
% \[
% u = u_1,\ v = \alpha^{-4/3} v_1,\ w = \alpha^{-8/3} w_1,\ m = m_1,\ \tau = \frac{3}{7}r^{7/3},
% \]
obtaining
\begin{equation}
\begin{alignedat}{2}
& u_{1\tau} & &= v_1\,, \\
& v_{1\tau} & &= w_1\,, \\
& w_{1\tau} & &= - \frac{(\alpha^2 + \overline{u}_1^2)^2}{D_4}m_1 + \alpha^{2/3}\frac{D_2}{D_4}(\alpha^2 + \overline{u}_1^2)v_1 - 5\frac{V}{D_4}(\alpha^2 + \overline{u}_1^2)^{3/2}\overline{u}_1 u_1 + \ldots\,,\\
& m_{1\tau} & &= \lambda \alpha^{4/3} u_1\,.
\end{alignedat}
\end{equation}
With this setup, we proceed to implement an $\varepsilon$-scaling with
\begin{equation}
    \begin{cases}
        u_1 = \varepsilon u_2\\
        v_1 = \varepsilon^{7/3} v_2\\
        w_1 = \varepsilon^{11/3} w_2\\
        m_1 = \varepsilon m_2\\
    \end{cases}
    \qquad
    \begin{cases}
        \overline{u}_1 = \varepsilon \overline{u}_2\\
        \overline{v}_1 = \varepsilon^{7/3} \overline{v}_2\\
        \overline{w}_1 = \varepsilon^{11/3} \overline{w}_2\\
        \overline{m}_1 = \varepsilon \overline{m}_2\\
        \alpha = \varepsilon \alpha_2\\
        \lambda = \varepsilon^{4/3} \lambda_2\\
        \tau = \varepsilon^{4/3} \sigma
    \end{cases}
\end{equation}
% $u_1 = \varepsilon u_2, v_1 = \varepsilon^{7/3} v_2, w_1 = \varepsilon^{11/3} w_2, m_1 = \varepsilon m_2, \overline{u}_1 = \varepsilon \overline{u}_2, \overline{v}_1 = \varepsilon^{7/3} \overline{v}_2, \overline{w}_1 = \varepsilon^{11/3} \overline{w}_2, \overline{m}_1 = \varepsilon \overline{m}_2,\alpha = \varepsilon \alpha_2, \lambda = \varepsilon^{4/3} \lambda_2$, and $\sigma$ 
such that $\partial_{\tau} = \varepsilon^{4/3}\partial_\sigma$. Substituting $\varepsilon=0$, we find the following system, now displaying all terms,
\begin{equation}\label{e:alllin}
\begin{alignedat}{2}
& u_{2\sigma} & &= v_2\,, \\
& v_{2\sigma} & &= w_2\,, \\
& w_{2\sigma} & &= F u_2 +   G v_2 + H w_2 + K m_2\,,\\
& m_{2\sigma} & &= \lambda_2 \alpha_2 u_2\,, \\
& \alpha_{2\sigma} & &= 0\,,
\end{alignedat}
\end{equation}
where
\begin{align*}
    F := &\frac{1}{D_4(\alpha_2^2 + \overline{u}_2^2)^3} \bigg(6D_4 \overline{u}_2 (3\alpha_2^2 - 5\overline{u}_2^2)\overline{v}_2^3 + 10D_4(\alpha_2^4 - \overline{u}_2^4)\overline{v}_2\overline{w}_2 -4\overline{m}_2\overline{u}_2(\alpha_2^2 + \overline{u}_2^2)^4 - 5 V \overline{u}_2 (\alpha_2^2 + \overline{u}_2^2)^{9/2}\bigg)\,,\\
    G := &\frac{1}{D_4 (\alpha_2^2 + \overline{u}_2^2)^3}\bigg(D_2\alpha_2^{2/3}(\alpha_2^2 + \overline{u}_2^2)^4 + 9D_4(\alpha_2^4 + 6\alpha_2^2\overline{u}_2^2 + 5\overline{u}_2^4)\overline{v}_2^2  + 10 D_4 \overline{u}_2 (\alpha_2^2 + \overline{u}_2^2) ^2\overline{w}_2\bigg)\,,\\
    H := &\frac{10\overline{u}_2\overline{v}_2}{(\alpha_2^2 + \overline{u}_2^2)}\,,\\
    K := &- \frac{(\alpha_2^2 + \overline{u}_2^2)^2}{D_4}\,.
\end{align*}
Writing the equation for the first four variables  $(u_2,v_2,w_2,m_2)$, in matrix form, we find the matrix
\[
\mathcal{A}(\lambda)=\begin{pNiceMatrix}
    0 & 1 & 0 & 0\\
    0 & 0 & 1 & 0\\
    F & G & H & K\\
    \lambda\mathrlap{_2} & 0 & 0 & 0 
\end{pNiceMatrix}.
\]
with nonnegative determinant $\lambda_2 (\alpha_2^2 + \overline{u}_2^2)^2/D_4$. Furthermore, recall that equilibrium solutions  $\overline{u}$  have  $\overline{v}_2 = \overline{w}_2 = 0$ and $(\alpha_2^2 + \overline{u}_2^2) = (\overline{m}_2/V)^2$. The characteristic polynomial of $\mathcal{A}(\lambda)$ then simplifies to 
\begin{align}
    \mathcal{P}(z) =& z^4 - \frac{D_2 D_4 \alpha_2^{2/3} \overline{m}_2^{14}}{V^{14}} z^2 + \frac{9D_4^2\overline{m}_2^{21}\sqrt{\overline{m}_2^2 - V^2 \alpha_2^2}}{V^{21}}z \nonumber\\
    &+ \frac{\lambda D_4^3 \overline{m}_2^{28}}{V^{28}}.
\end{align}
Key here is that the expression only depends on $\alpha_2$ and the constant $\overline{m}_2=\Omega_*=\omega/\eps$, and $\alpha_{2\sigma}=-\eps \alpha_2^{10/3}$ varies slowly for $0<\eps\ll1$.

As a result, we have an equation of the form 
\[
W_\sigma=\mathcal{A}(\eps \sigma;\lambda)W,
\]
equipped with ``clamped'' boundary conditions that force the first two components to vanish. Such boundary conditions are always transverse to stable and unstable eigenspaces to $\mathcal{A}$ due to the higher-order nature of the equation. The results in \cite{CRS} demonstrate that eigenvalues then accumulate at locations of the absolute spectrum of $\mathcal{A}(\tilde{\sigma};\lambda)$. In fact, transversality of the boundary subspaces implies that eigenvalues only accumulate at such locations. The absolute spectrum, in turn, is given by values of $\lambda$ and $\tilde{\sigma}$ where eigenvalues of $\mathcal{A}(\tilde{\sigma};\lambda)$ cannot be separated evenly by real part, that is, there does not exist a splitting of eigenvalues  $\nu_{1,2,3,4}$ (repeated by multiplicity) of the form 
\[
\Re\,\nu_{1,2}<\Re\,\nu_{3,4}.
\]
For fixed $\tilde{\sigma}$, the values of $\lambda$ where this splitting fails come in algebraic curves which terminate at double roots, $\nu_2=\nu_3$, say. In order to determine the most unstable point of the absolute spectrum, we assume that it is given by such a double root. These double roots can in fact be located explicitly. For this,
% Notice that this is only in terms of $D_2, D_4, V, \overline{m}_2$ (which are constant), and $\alpha_2$. 
consider $\mathcal{P}(z)= 0$ and set  $z=\rmi k$, with $k$ possibly complex, 
% as an ansatz ($k^*$ possibly complex) with the goal of finding eigenvalue $ik^*$ that solves the dispersion relation:
\begin{align}\label{e:disp}     
    \lambda =& -\frac{kV^7}{D_4^3 \overline{m}_2^{28}}\bigg(k^3V^{21} + kV^7 D_2 D_4 \alpha_2^{2/3}\overline{m}_2^{14} \nonumber\\
    &+ 9\rmi D_4^2\overline{m}_2^{21}\sqrt{\overline{m}_2^2 - V^2 \alpha_2^2}\bigg)\,.
\end{align}
Rather than searching for double roots $k$ directly in this equation, we notice that the equation is obtained via Fourier-Laplace transform, $A(t,r)=\rme^{\lambda t + \rmi k r}$, as the dispersion relation to 
\begin{equation}\label{e:4thpar}
A_t = aA_{rrrr} + b(\alpha_2)A_{rr} + c(\alpha_2)A_r,
\end{equation}
with
\begin{equation}
    \begin{aligned}
        a & = -\frac{V^{28}}{D_4^3\overline{m}_2^{28}}\,,\,
        b & = \frac{V^{14}D_2\alpha_2^{2/3}}{D_4^2\overline{m}_2^{14}}\,,\,
        c & = -\frac{9V^7\sqrt{\overline{m}_2^2 - V^2\alpha_2^2}}{D_4\overline{m}_2^7}\,.
    \end{aligned}
\end{equation}
Note that $a$, $c < 0$, and that $b<0$ follows from $D_2<0$. 
This fourth-order diffusion equation and its instabilities have been analyzed in the context of front invasion problems in \cite[\S2.11.1]{vansaarloos}; see also \cite[\S2.7]{avery2025selectionmechanismsinvasion} form. Summarizing the calculations there, we look for double roots as endpoints of the absolute spectrum crossing the imaginary axis and hence inducing an instability when $|c|<c_\mathrm{lin}$.  Here, $c_\mathrm{lin}$ is the linear spreading speed to the equation without advection, $A_t=a A_{rrrr}+b(\alpha_2)A_{rr}$. The spreading speed of instabilities in equations of this form is given through, in the notation of \cite{vansaarloos} where $\lambda=\rmi\omega$, the solution of the complex equation
\begin{equation}\label{eqn:wim-speed}
    \frac{\rmd \omega}{\rmd k} \bigg\rvert_{k_*} = \frac{\Im(\omega(k_*))}{\Im(k_*)}\,,
\end{equation}
where $\omega = \rmi(a k^4 + b k^2)$.
Since the right-hand side in \eqref{eqn:wim-speed} is real, we have 
\begin{equation}\label{eqn:zero-restr}
    \Im\left(\frac{{\rmd }\omega}{{\rmd } k} \bigg\rvert_{k_*}\right) = 0\,.
\end{equation}
We now solve for $k_*$. Setting $k_* = x + \rmi y$,~(\ref{eqn:zero-restr}) gives  $x = 0$ or $x = \sqrt{\frac{b + 6ay^2}{2a}}$. The case  $x=0$, implies that $y=0$ or $y=\sqrt{-\frac{b}{3a}}$, the latter of which breaks down when $D_2 < 0$, as $y\in\mathbb{R}$, and is thus irrelevant for us. Next, set $x = \sqrt{\frac{b + 6ay^2}{2a}}$.  We substitute into~(\ref{eqn:wim-speed}) to obtain the following set of solutions to $k_*$, 
\begin{equation}
    \begin{aligned}
        x & = \pm \sqrt{\frac{b(3 + \sqrt{7})}{8a}}\,,\qquad
        y & = \pm \sqrt{\frac{b}{24a}(\sqrt{7}-1)}\,.
    \end{aligned}
\end{equation}
% ask what the signs on x and y signify and what choosing the sign of y signifies
Choosing the negative root for $y$, we obtain that the spreading speed is
\begin{equation}
    c_\text{lin} := \frac{\Im(\omega(k_*))}{\Im(k_*)} = \frac{b(5 + \sqrt{7})}{9}\sqrt{\frac{b(1+\sqrt{7})}{a}}.
\end{equation}
Equating this with the term $c(\alpha_2)$ yields
% 
% Bringing back the linear term $c$ in the dispersion equation, we set our linear term $c$ equal to the spreading speed. Heuristically, this is the equivalent of pausing the system, as the wave movement from the linear term matches the spreading rate. Thus, solving the equation $s_\text{lin} = c$ for $\alpha_2$ will give an explicit expression for the onset of absolute growth and instability for a fixed $D_2 < 0$. Substituting both sides of the expression gives
\[\frac{(5+\sqrt{7})V^7D_2\alpha_2}{9D_4^2\overline{m}_2^7}\sqrt{-D_2D_4(1 + \sqrt{7})} = -\frac{9V^7 \sqrt{\overline{m}_2^2 - V^2\alpha_2^2}}{D_4 \overline{m}_2^7}\,,\]
and, after some elementary algebra 
\begin{equation}
    \alpha_{2,\mathrm{crit}} = \frac{81\overline{m}_2 D_4}{\sqrt{81^2D_4^2V^2 - 6D_2^3D_4(17+7\sqrt{7})}}\,, \qquad \overline{m}_2=\Omega_*.
\end{equation}
Substituting the leading-order expansion $\Omega_*=V/(R_\mathrm{i}\sin(\vartheta_\mathrm{i}))$, we find that $R_\mathrm{i,crit}=1/\alpha_\mathrm{i,crit}>R_\mathrm{i}$ precisely when 
\begin{equation}
D_2<D_{2,\mathrm{crit}}=-\sqrt[3]{\frac{81}{4}(7\sqrt{7}-17)D_4 V^2 \cot^2(\theta_\mathrm{i})}\leq 0.
\label{e:alphac}
\end{equation}
Clearly, $D_{2,\mathrm{crit}}=0$ when $\vartheta_\mathrm{i}=\pi/2$,  reflecting the absence of transport at the boundary when the filament is perpendicular to the boundary. More generally, for unstable $D_2<D_{2,\mathrm{crit}}$, we find a region near the core where the system is absolutely unstable. This reflects the observation that the increased curvature of the spiral near the core increases the pointwise growth rate of perturbations by reducing the speed of outward transport. 
% Note that $D_2 < 0$ guarantees the existence of a corresponding real $\alpha_2$. 
% In general, we have $c_\mathrm{lin}>c$ when $\alpha_2>\alpha_{2,\mathrm{crit}}$, which indicates that the region of absolute instability is given by a region near the core where $\alpha_2>\alpha_{2,\mathrm{crit}}$. 
% In particular, as we decrease $D_2$, instabilities and unstable eigenvalues first emerge when $\alpha_{2,\mathrm{crit}}=\alpha_{2,\mathrm{i}}$, given by the core radius, that is, when 
% \[
% R_\mathrm{i}=\hat{R}_\mathrm{i}/\eps=1/\alpha_{2,\mathrm{crit}}.
% \]
The relation \eqref{e:alphac} has some immediate monotonicity properties triggering eventually instabilities for fixed $D_2<0$, as:
\begin{enumerate}
    \item  $\Omega_*$ is decreased, by for instance choosing a more obtuse (increased) contact angle effects a decrease in  $|D_{2,\mathrm{crit}}|$;
    \item $V$ is decreased, which also decreases $\Omega_*$ linearly in $V$; 
    \item$D_4$ is decreased.
\end{enumerate}
Our results concern the point spectrum, only. The convective nature of the instability together with the possibility of singularity formation suggests that even in the absence of the unstable point spectrum, the linearly stable spirals are difficult to observe due to the unstable pseudospectrum and exponentially small basin boundaries; see for instance \cite{ssbasin} for an analysis in a simple model problem. 
% fig:scaling-direct-simulation
\begin{figure*}[h!]
    \centering
    \includegraphics[trim = {1.5cm, .65cm, 1.5cm, 1.2cm}, clip,width=0.8\linewidth]{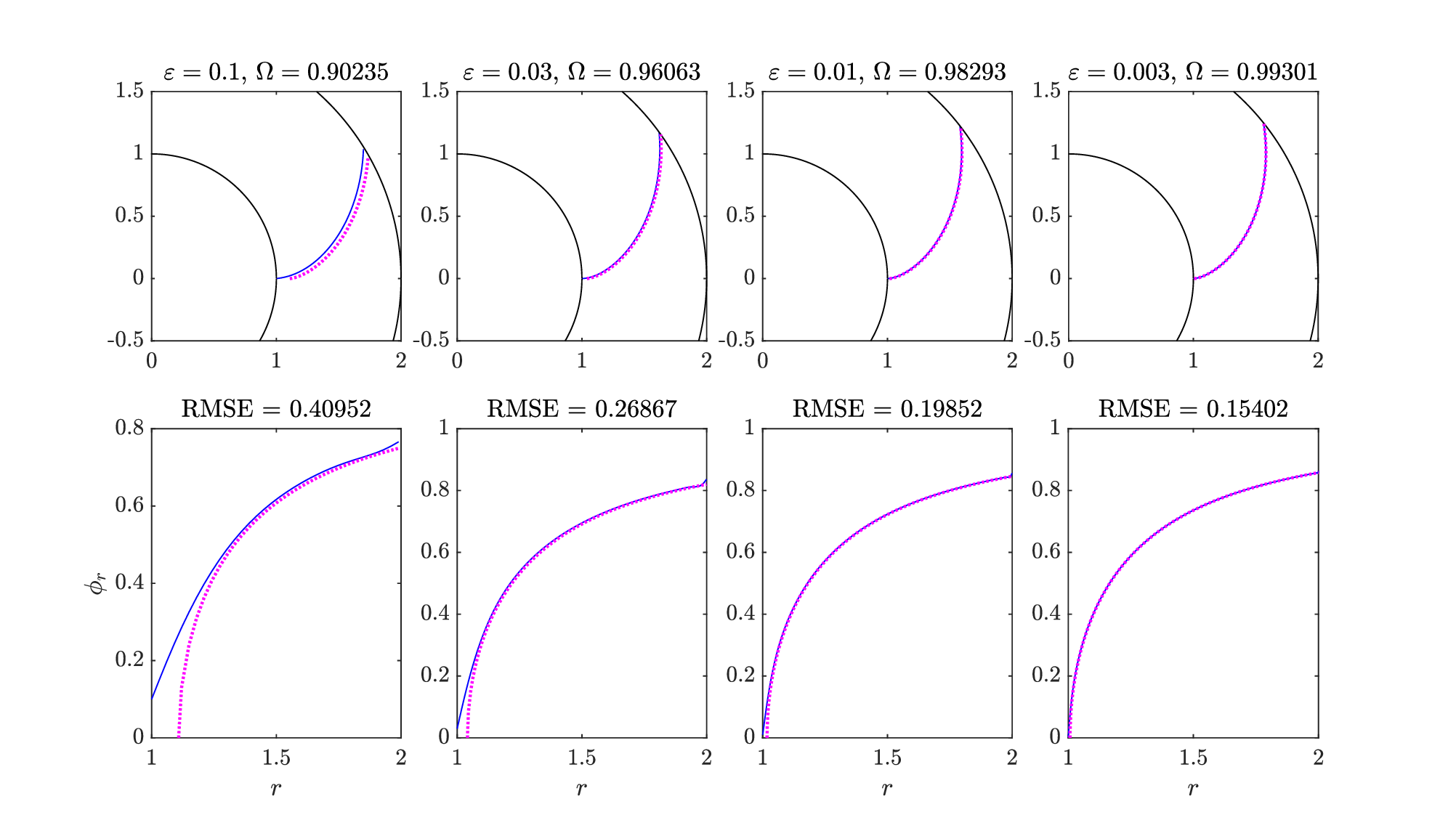}
    \caption{Equilibrium solution curves from direct simulations (top row) on a domain with $R_\mathrm{i}=1$, $R_\mathrm{i}=2$, $\vartheta_\mathrm{i}=\vartheta_\mathrm{o}=\pi/2$,  with $V=1$, $D_2=0.1\eps$, $D_4=\eps^3$, for values of $\varepsilon=0.1,0.03,0.01,0.003$ with measured values of $\Omega$ converging linearly in $\eps$ to the eikonal limit $\Omega=1$. 
   Computed profiles (blue) compared with the predicted, eikonal spiral profile (pink). Also shown is the graph of $\phi_r$ (bottom row) with discrepancies to the eikonal reference most dominant near the boundaries due to boundary layers at the critical contact angle; see also \cite{nanth}.}
    \label{fig:scaling-direct-simulation}
\end{figure*}
% \begin{figure}
%     \centering
%     \includegraphics[width=.7\linewidth]{Figures/order_of_convergence.eps}
%     \caption{Convergence of observed rotation rate with smaller $\varepsilon$; parameters as in Fig.~\ref{fig:pertConv} for various values of $\sigma=D_2/D_4$.}
%     \label{fig:order-of-convergence}
% \end{figure}

In the case of an unbounded domain, we suspect that the point spectrum induced by the instability at finite $\alpha_2$ is still present. Instability at larger radii is convective and perturbations appear to lead to singularities at sufficiently large distances from the core. We emphasize however that the growth rate of the instability decreases as $r\to\infty$ due to the factor $\alpha$ in the unstable coefficient $b(\alpha_2)$ in \eqref{e:4thpar}. More directly, at $r=\infty$, $\alpha_2=0$, the linearization has $b=0$ and is therefore (marginally) stable, regardless of the sign of $D_2$. In other words, investigating the essential spectrum by studying the linearization at $r=\infty$, only, misses the instability caused by negative surface tension, $D_2<0$, completely, reflecting the observation in \cite{sandstede2021spiral} that spectral mapping theorems do not hold for transverse instabilities of spirals.

In the next section, we show some numerical results corroborating and illustrating our predictions in Section~\ref{s:3} and Section~\ref{s:4}.

%%%%%%%%%%%%%%%%%%%%%%%%%%%%%%%%%%%%%%%%%%%%%%%%%%%%%%
\section{Numerical analysis and morphology of instabilities}\label{s:5}

We describe numerical simulations that corroborate our asymptotics, exhibit morphology of instabilities, and point to several open questions. We discretized \eqref{eq:2.6} using a second-order finite-difference approximation based on centered second-order stencils for the spatial derivatives. For time-stepping, we use a semi-implicit Euler scheme where  $\Phi_{rrrr}$ is treated implicitly and all other occurrences of $\Phi$ and its derivatives are treated explicitly. Boundary conditions are computed using ghost points.  Grid sizes are usually $dr=0.05$, scaled to reflect the size of $D_4$. Typical time steps are $dt=10^{-3}$ due to the strong nonlinearities.   

We first confirmed convergence to the eikonal solution when $R_\mathrm{i}$ is large, or, equivalently, $D_2,D_4\ll V=\rmO(1)$; see Figure~\ref{fig:scaling-direct-simulation}. We did see the predicted linear convergence in $\varepsilon$ of profiles and frequencies. 

% The convergence in fact depends on the value of $D_2$ as shown in Fig.~\ref{fig:order-of-convergence}. Interestingly, convergence is somewhat faster when $D_2$ is smaller, which we attribute to a weakening of the boundary layer that was studied in \cite{nanth}. 
% fig:pertConv
\begin{figure*}[h!]
    \centering
    \includegraphics[trim = {4.5cm, .7cm, 4.5cm, 0.7cm}, clip, width=0.8\linewidth]{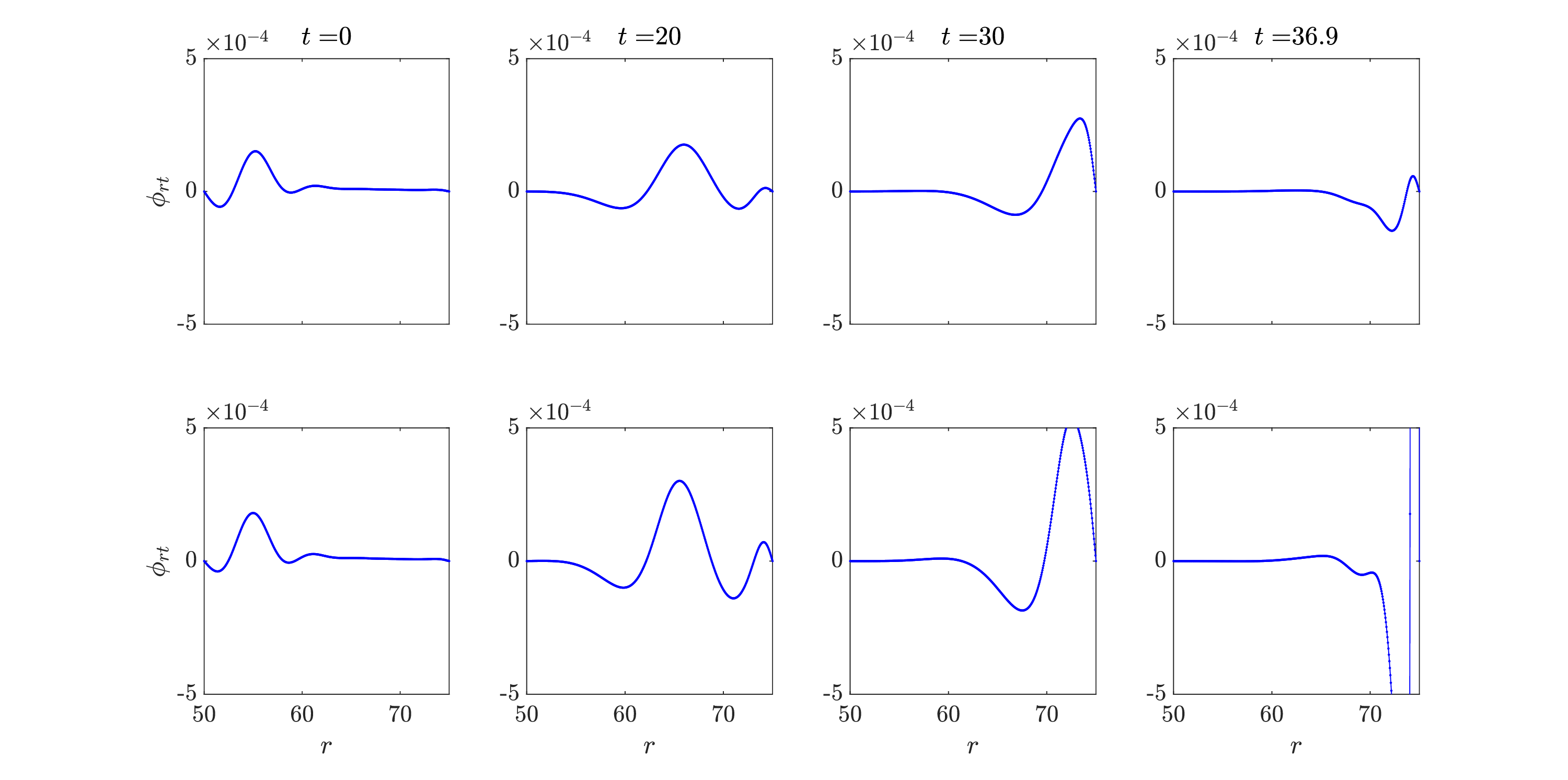}
    \caption{An initial Gaussian perturbation is advected to the outer boundary. Both time series are for $R_\mathrm{i} = 50, R_\mathrm{o} = 75$, $\vartheta_\mathrm{i}=\pi/2-0.1$, $D_4=V=1$, so that $D_{2,\mathrm{crit}}\sim-0.67$ according to \eqref{e:alphac}. 
    For $D_2 = -0.5$ (top) the perturbation decays as it appears to  pass through the boundary. Closer to instability, $D_2 = -0.6$ (bottom), the perturbation exhibits faster growth while it is advected and blows up once it reaches the boundary. 
    }
    \label{fig:pertConv}
\end{figure*}

We measured frequencies in direct simulations but also compared with frequencies determined using a Newton-method for the time-stepping scheme with phase condition $\int\phi=0$ and Lagrange multiplier $\omega$. We compared with multiple shooting methods for the ODE \eqref{eq:3.1}. While more straightforward to implement, are quite ill-behaved for small $\varepsilon$ due to the superexponential growth of perturbations to the initial value, despite the use of high-precision arithmetic. 

We also studied the transition to instability numerically. In these simulations, we used compatible boundary conditions with vanishing curvature and with contact angles $\vartheta_\mathrm{i}=\pi/2 - 10^{-3}$ and $\vartheta_\mathrm{o}$ according to \eqref{e:comp}. Instabilities always involve strong transport toward the boundary, which is convective at first, that is,  the maximum value of the perturbation grows exponentially as it is advected to the outer boundary where it eventually dies out. Past the transition to absolute instability, perturbations also grow pointwise. Figure~\ref{fig:pertConv} 
displays the transition from the eventual decay of perturbations to the eventual development of a singularity, as $D_2<0$ is decreased. We are plotting the second derivative $\Phi_{rt}$ to measure the relaxation of $\Phi_t$ while also eliminating the spatially constant rotation $\Phi_t=\omega$.  The instability appears to set in slightly before the threshold $D_{2,\mathrm{crit}}$. This may be due to $\rmO(\eps)$-corrections, due to the subcritical nature of the instability with basins of attraction exponentially small in the radius \cite{ssbasin}, or even due to the presence of boundary modes. 
% fig:bumpMove
\begin{figure*}[t!]
    \centering
    \includegraphics[trim={5.3cm 7.8cm 4.2cm 7.2cm},clip, width=0.8\linewidth]{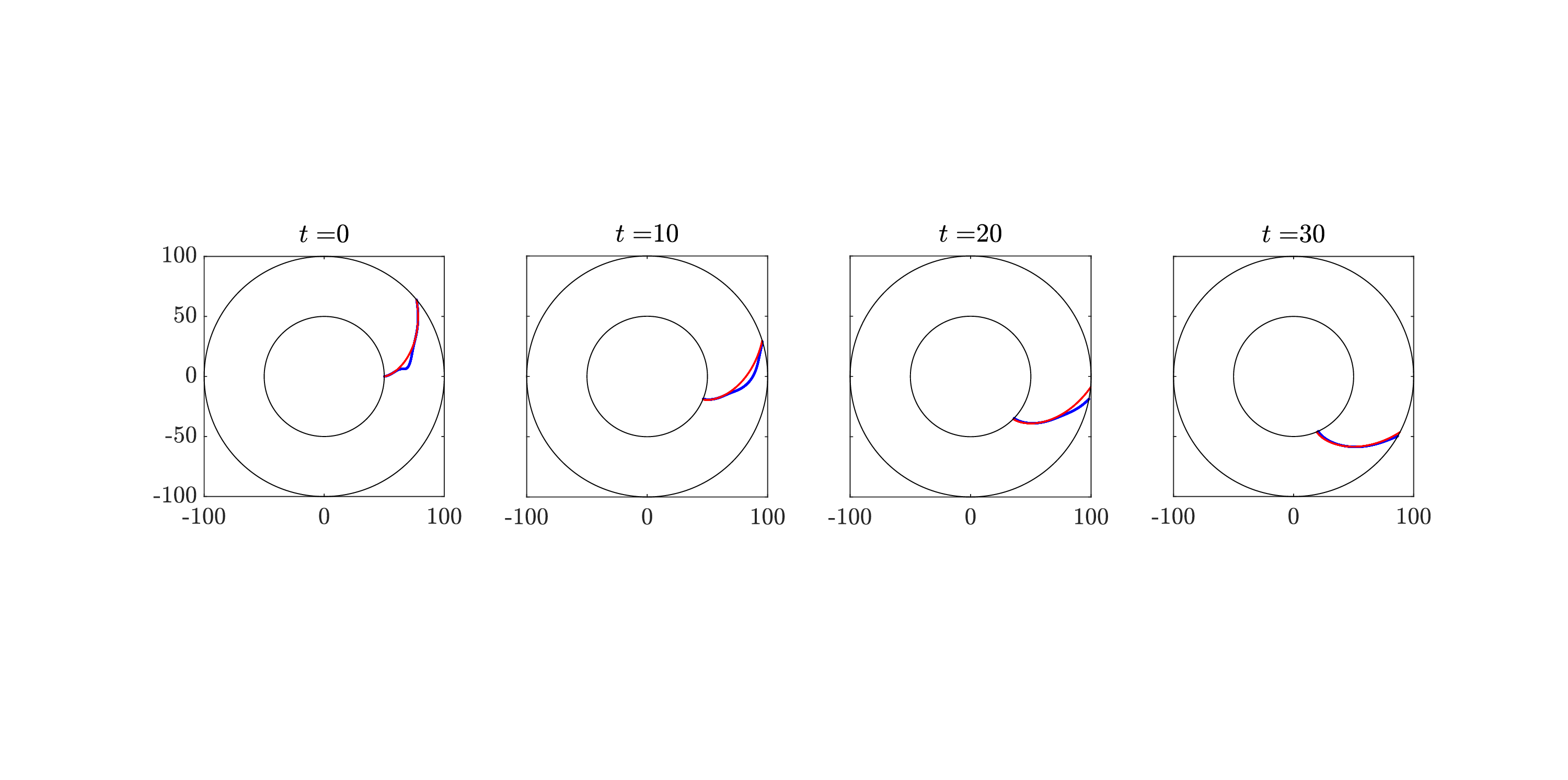}\\[0.2in]
     \includegraphics[trim={5.2cm, 1.cm, 4.5cm, 1.1cm}, clip, width=0.8\linewidth]{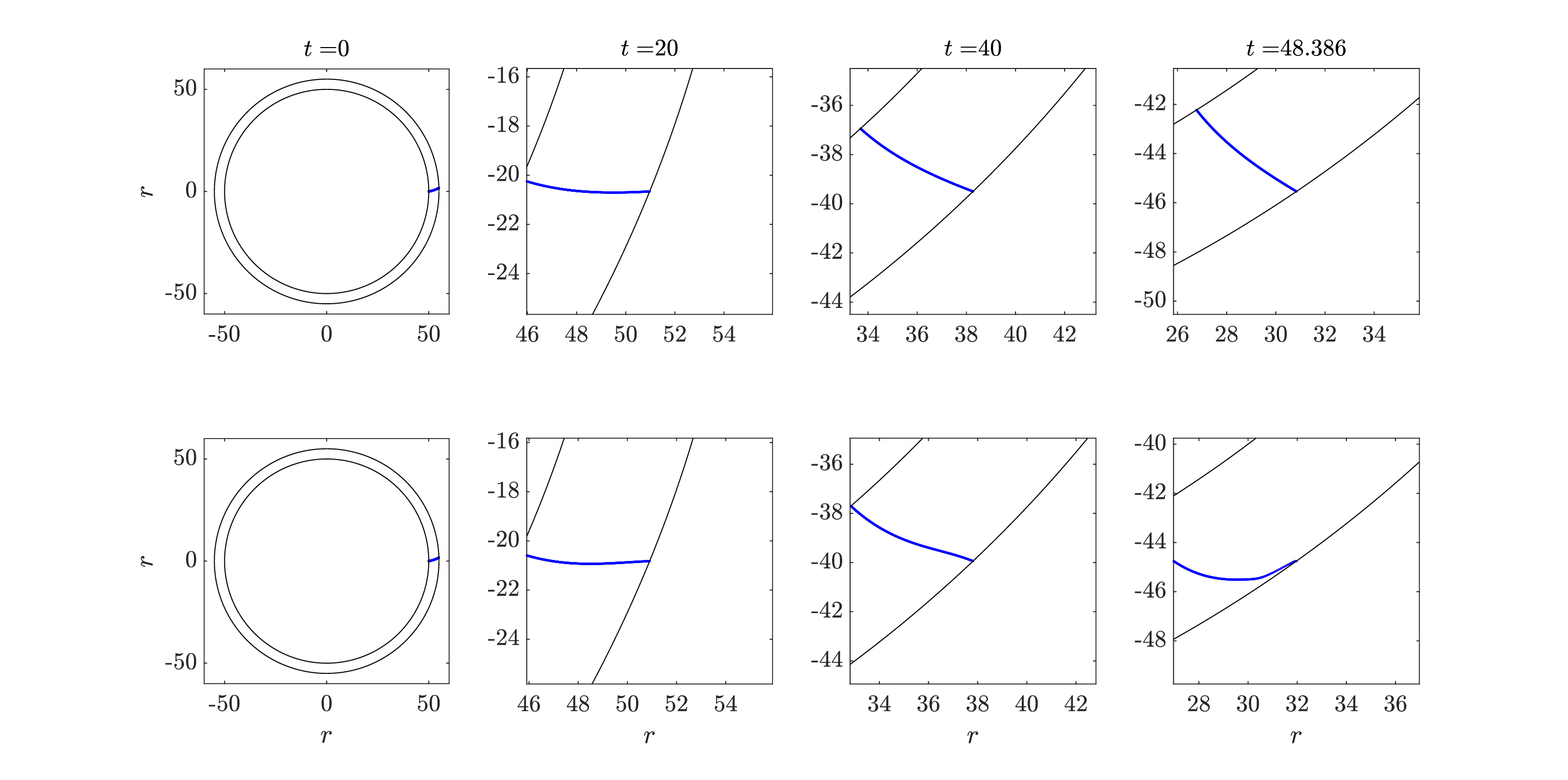}
    \caption{Top: Evolution of eikonal solution (red) with small Gaussian perturbation (blue) for parameters $V = 1, D_2 = 1, D_4 = 1, R_\mathrm{i} = 50, R_\mathrm{o} = 150$ (note only part of the domain is shown), which is firmly in the stable regime. The perturbation is advected towards and through the outer boundary while also diffusing.
    Middle and bottom row: Evolution of a perturbed eikonal solution with  $D_2 = -0.5$ (top) and $D_2 =-0.6$ (bottom) is shown in a thin annulus $r\in(50,55)$, illustrating profiles during convective stability and absolute instability, respectively, with blowup near the boundary in the bottom right.
    % for the rotating spiral solution we study, and initial condition of the eikonal solution perturbed by a small gaussian (blue) evolves. A fitted rotation of the eikonal solution (red) is plotted alongside this solution to better visual the evolution of the perturbation. 
    % \color{red}{There is a bit of a lie in this, In order to see the gaussian be advected, one has to simulate on a pretty big domain because as soon as a macroscopic perturbation gets near the boundary we get blow up. So I simulate with $R_\mathrm{i} = 50, R_0 = 150$, but only plot for $ R_\mathrm{i} = 50, R_0 = 100$. Unsure how/ if I should address this in the caption}
    }
    \label{fig:bumpMove}
\end{figure*}

Both stability and the formation of a singularity near the outer boundary are also shown in Figure~\ref{fig:bumpMove} for the actual spiral profile. Stability is seen through both diffusive decay and advection towards the boundary; instability as the formation of a corner on the outer boundary of the solution, $\phi$. Time series near criticality show how the perturbation eventually blows up near the boundary. Choosing a thin annulus, we avoided some of the difficulties with the subcritical nature of the instability and the ensuing very small basins of attraction near criticality. Indeed,   for both convective and absolute instabilities, the maximum value of the perturbation grows exponentially as it is advected to the outer boundary. Thus, if one wants to witness instability before a singularity develops, anywhere other than at the outer boundary, a rather thin annulus is needed. Even then the growth is most prevalent on the outer boundary.

We also demonstrate that the instability does indeed transition to absolute growth by plotting the amplitude of the perturbation,  summed in an inner core annulus, over time in Figure~\ref{fig:growthNearCore}.

% \begin{figure}[h!]
%     \centering
%     \includegraphics[trim={4.5cm, 1.cm, 4.5cm, 1.1cm}, clip, width=1\linewidth]{Figures/CornerForm.eps}
%     \caption{Evolution of a perturbed eikonal solution with different values of $D_2$, with   $D_2 = -0.5$ (top) and $D_2 =-0.6$ (bottom) is shown in a thin annulus $R_\mathrm{i} = 50$ and $ R_\mathrm{o} = 55$, illustrating profiles during convective stability and absolute instability, respectively, with blowup near the boundary in the bottom right.}
%     \label{fig:cornerForm}
% \end{figure}

We conclude this section noticing that for finite, not small, $\varepsilon$, one observes saddle-node bifurcations rather than the oscillatory convective and absolute instabilities we analyzed here. Some numerical observations of this saddle-node instability are shown in Figure~\ref{fig:inset-bifurcation}. It would be interesting to understand how the Hopf instabilities observed here transition into saddle-node bifurcations for finite $\varepsilon$.

% describe the numerical simulations performed to investigate the transition from convective to absolute instability and instability seen at the boundary. That is, those needed to generate figures \ref{fig:pertConv},\ref{fig:bumpMove}, 

% * comment on the infeasibility of shooting, ...

% * $\Omega$ convergence

% * Figure. 7?? convection, then instability of boundary layer??

% * discuss a bit the saddle node.

% fig:growthNearCore
\begin{figure*}
    \centering
    \includegraphics[trim={2cm 4cm 2cm 6.2cm},clip,width=0.8\linewidth]{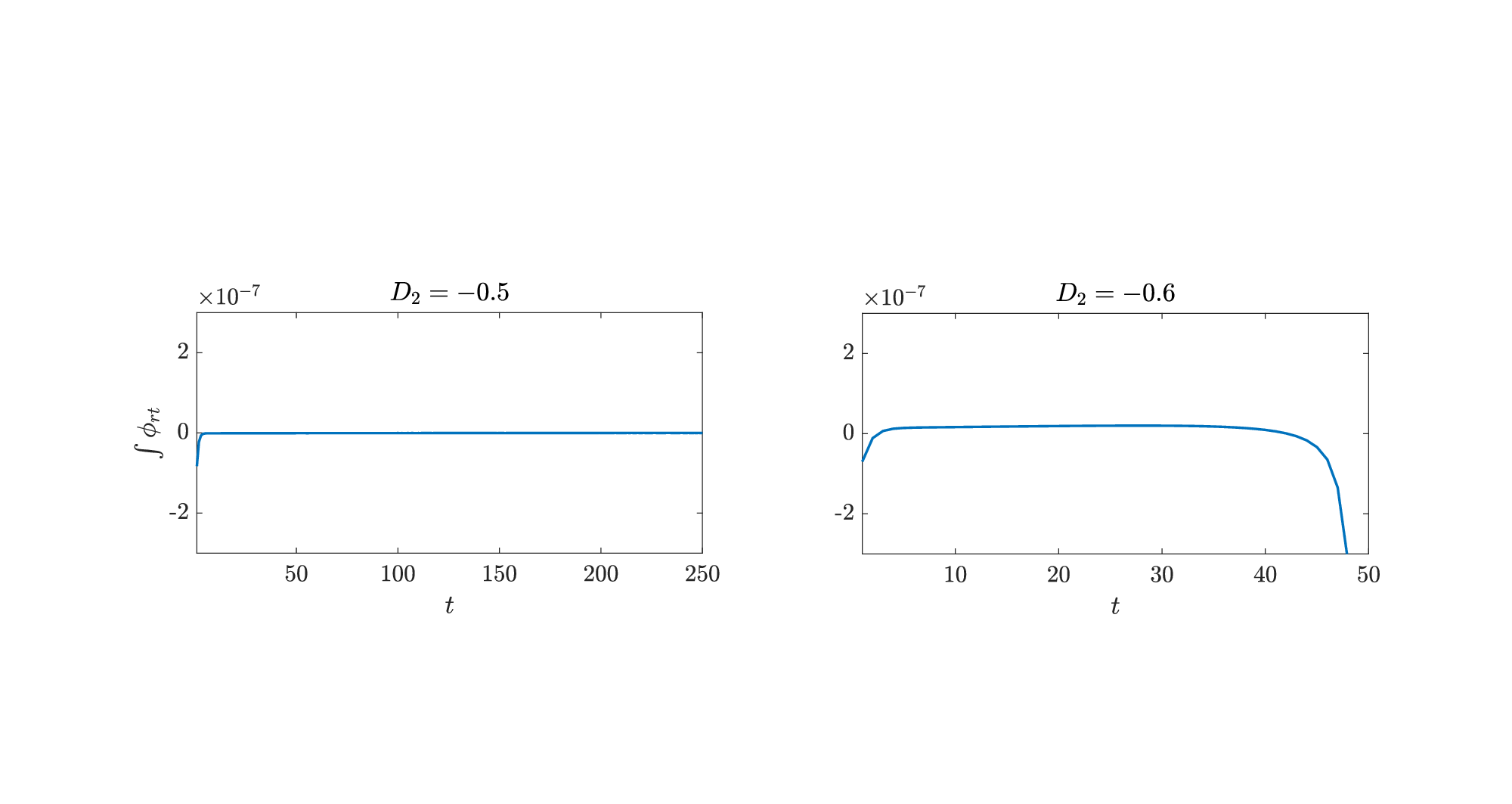}
    \caption{We simulate \eqref{eq:2.6} with initial condition chosen as exactly the eikonal solution and a small perturbation added near the inner boundary; parameters as in Figure~\ref{fig:bumpMove}. The integral of $\phi_{rt}$ over the inner quarter of the simulated annulus is plotted against time as a measure of the pointwise stability of the solution $\phi$ near the core.  For $D_2 = -0.5$, after some initial fluctuation, the integral value settles near zero suggesting stability near the core. For $D_2 = -0.6$ the same initial condition exhibits growth near the core and blows up in simulation, suggesting instability near the core, thus corroborating the theoretically predicted transition from convective to absolute instability.}
    \label{fig:growthNearCore}
\end{figure*}
% fig:inset-bifurcation
\begin{figure*}
    \centering
    \includegraphics[trim={2cm 1.cm 2cm 2.cm},clip,width=.8\linewidth]{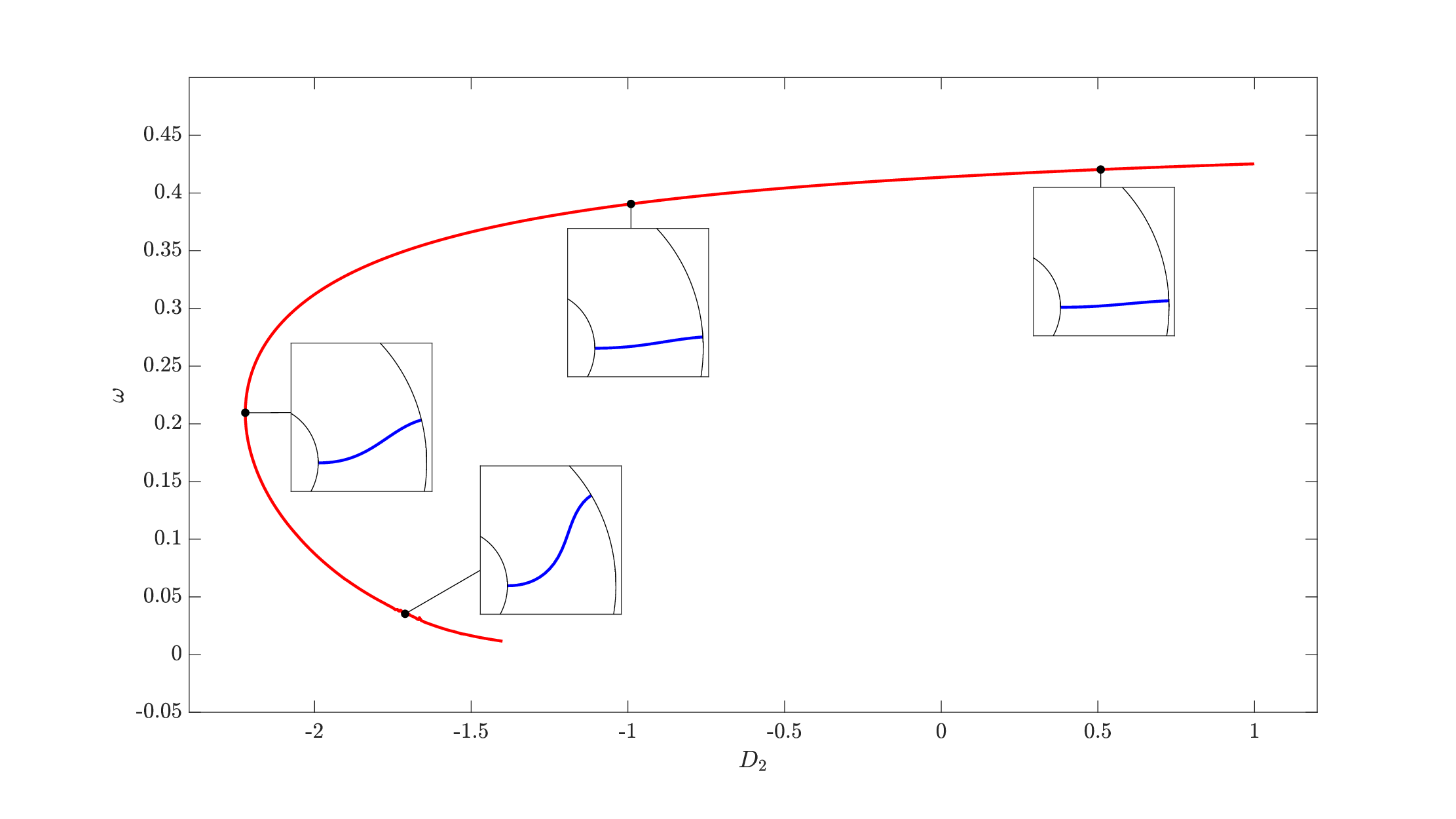}
    \caption{Numerical secant continuation in a shooting method was used to follow rigidly rotating spiral waves in a finite annular domain with parameters $D_4=V=R_\mathrm{i}=1$, $R_\mathrm{o}=3$, $\vartheta_\mathrm{i}=\vartheta_\mathrm{o}=\pi/2$, and with a starting value $D_2=1$. Insets display the computed spiral at selected points.}
    \label{fig:inset-bifurcation}
\end{figure*}

% \begin{table}
%     \centering
%     \begin{tabular}{@{}lrrrrrrrrr@{}}
%     \toprule
%     $\boldsymbol\varepsilon$ & .3 & .2 & .1 & .03 & .02 & .01 & .003 & .002 & .001 \\
%     \midrule
%     $\boldsymbol\Omega_d$ & .64158 & .69132 & .78179 & .90421 & .92943 & .95955 & .98701 & .992 & .99765 \\
%     \midrule
%     $\boldsymbol\Omega_s$ & .64081 & .72496 & .83677 &  \\
%     \bottomrule
%     \end{tabular}
%     \caption{$\varepsilon$ values against equilibrium $\Omega$ values, $\Omega_d$ for direct simulation and $\Omega_s$ for shooting. Shooting fails for $\varepsilon<0.055$, which gives $\omega = 0.89912$}
%     \label{tab:eps-omega}
% \end{table}

%%%%%%%%%%%%%%%%%%%%%%%%%%%%%%%%%%%%%%%%%%%%%%%%%%%%%%
\section{Discussion and open problems}\label{s:dis}
%%%%%%%%%%%%%%%%%%%%%%%%%%%%%%%%%%%%%%%%%%%%%%%%%%%%%%

Summarizing, we established the existence of rigidly rotating spirals in a model for geometric curve evolution by studying a singular perturbation from a somewhat explicit eikonal flow limit. Most interestingly, our existence result is valid in regions where one expects the spiral to be unstable due to a negative line tension term. We analyze this instability using again perturbation theory from the singular eikonal limit, exhibiting an eigenvalue problem with slowly varying coefficients. As a consequence, we predict eigenvalue clusters destabilizing the spiral wave at a critical parameter value where the instability changes from convective to absolute. The instability always originates near the core. In this sense, it \emph{precedes} a potential absolute instability at spatial infinity. 

The results can somewhat formally be compared to the analysis in \cite{ss_curvature}, where point spectrum clusters in spiral spectra more generally were analyzed. The authors there found eigenvalues accumulating at branch points of the dispersion relation of wave trains due to curvature effects. Since those branch points mark pointwise growth rates of perturbations of the wave trains, one can interpret the results there as determining whether curvature effects enhance or weaken pointwise growth. The case of enhancing pointwise growth is analogous to the situation we encounter here. It would be interesting to understand whether the possibility of curvature effects weakening pointwise growth can at all occur in the type of geometric model we considered here, although pointwise growth at infinity is invisible in the spectrum, due to a hypothetical branch point at $\pm\rmi\infty$. Analyzing such accumulation of spectra in a reaction-diffusion model, both theoretically and numerically may shed light on this question. 

In a different direction, it would be interesting to, at least numerically, study boundary layers and their instability. In fact, it is often difficult to determine whether in direct simulations perturbations blow up at the boundary, just because their amplitude is largest there, or because of an intrinsic instability nested right at the boundary. The difficulty of boundary layers is of course also reflected in the statement of our main results, which assumes ``compatible'' boundary conditions, excluding a quite natural case of perpendicular contact angles. An analysis of this critical case and the associated singular perturbation problem appears to pose some quite interesting challenges. 

%%%%%%%%%%%%%%%%%%%%%%%%%%%%%%%%%%%%%%%%%%%%%%%%%%%%%%

%%%%%%%%%%%%%%%%%%%%%%%%%%%%%%%%%%%%%%%%%%%%%%%%%%%%%%
% Literature
%%%%%%%%%%%%%%%%%%%%%%%%%%%%%%%%%%%%%%%%%%%%%%%%%%%%%%
% \bibliographystyle{abbrv}
% \bibliography{p2}

\bibliographystyle{abbrv}

\end{document}